%% file: main.tex
\documentclass[a4paper,10pt,twoside]{article}

\usepackage[utf8]{inputenc} 
\usepackage[english]{babel}
\usepackage[english=american]{csquotes}
\MakeOuterQuote{"}

\usepackage{authblk}

\usepackage{amsmath,amsthm,amssymb,mathtools,enumitem,physics} 
\usepackage[only,boxslash]{stmaryrd} 
\usepackage{hyperref} 
\usepackage{geometry} 

\usepackage{graphicx}
\graphicspath{{images/}}

\newcommand{\im}{\mathrm{i}}

\newcommand{\R}{\mathbb{R}}
\newcommand{\C}{\mathbb{C}}
\newcommand{\defeq}{\coloneqq}
\newcommand{\tens}{\otimes}
\DeclareMathOperator{\id}{\mathbb{I}}
\newcommand{\xd}{\mathrm{d}}
\newcommand{\xD}{\mathcal{D}}
\newcommand{\cH}{\mathcal{H}}

\newcommand{\coh}{\mathsf{K}}
\newcommand{\ncoh}{\mathsf{k}}
\newcommand{\one}{\mathbf{1}}
\newcommand{\cM}{\mathcal{M}}
\newcommand{\cB}{\mathcal{B}}
\newcommand{\cT}{\mathcal{T}}
\newcommand{\comp}{\diamond}
\newcommand{\cA}{\mathcal{A}}

\newcommand{\tord}{\mathbf{T}}
\newcommand{\Lint}{\mathrm{int}}
\newcommand{\Lext}{\mathrm{ext}}
\newcommand{\wuu}{w_{\uparrow\uparrow}}
\newcommand{\wdd}{w_{\downarrow\downarrow}}
\newcommand{\wud}{w_{\uparrow\downarrow}}
\newcommand{\wdu}{w_{\downarrow\uparrow}}
\newcommand{\wtuu}{\mathtt{w}_{\uparrow\uparrow}}
\newcommand{\wtud}{\mathtt{w}_{\uparrow\downarrow}}
\newcommand{\wtdu}{\mathtt{w}_{\downarrow\uparrow}}

\newcommand{\cP}{\mathcal{P}}
\newcommand{\discard}{\vcenter{\hbox{\includegraphics[width=1em]{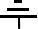}}}}

\newcommand{\vac}{\mathbf{0}}
\newcommand{\np}{\boxslash}
\newcommand{\lf}{\uparrow}
\newcommand{\lp}{\downarrow}
\newcommand{\cN}{\mathcal{N}}
\newcommand{\cI}{\mathcal{I}}

\newcommand{\sdd}{\Sigma_{\downarrow\downarrow}}
\newcommand{\sdu}{\Sigma_{\downarrow\uparrow}}
\newcommand{\sud}{\Sigma_{\uparrow\downarrow}}
\newcommand{\suu}{\Sigma_{\uparrow\uparrow}}

\newcommand{\wh}{r}
\newcommand{\wth}{\mathtt{r}}
\newcommand{\wta}{\mathtt{w}_{\mathrm{adv}}}
\newcommand{\wtr}{\mathtt{w}_{\mathrm{ret}}}
\renewcommand{\wr}{w_{\mathrm{ret}}}
\newcommand{\wa}{w_{\mathrm{adv}}}

\theoremstyle{definition}
\newtheorem{dfn}{Definition}[section]

\theoremstyle{plain}
\newtheorem{lem}[dfn]{Lemma}
\newtheorem{prop}[dfn]{Proposition}
\newtheorem{thm}[dfn]{Theorem}

\allowdisplaybreaks

\begin{document}

\input{titlepage}


\input{intro}
\input{causaltransp}

\input{quantization}

\input{spacetime}

\input{multiobs}

\input{causalcor}

\input{outlook}

\subsection*{Acknowledgments}

I would like to thank Adamantia Zampeli, Maria Eftychia Papageorgiou, Robin Simmons, Jan Mandrysch and Miguel Navascués for valuable discussions that have contributed to shaping this research. I am grateful for the support of the Institute for Quantum Optics and Quantum Information, Vienna, where most of this work was carried out during a sabbatical. This work was also partially supported by UNAM-PAPIIT project grant IN106422 and UNAM-PASPA-DGAPA.

\numberwithin{equation}{section}

\appendix

\input{mathspec}

\input{mathst}

\input{pmultiobs}

\newcommand{\eprint}[1]{\href{https://arxiv.org/abs/#1}{#1}}
\bibliographystyle{stdnodoi} 
\bibliography{stdrefsb}
\end{document}

%% file: titlepage.tex

\begin{titlepage}
\title{\textbf{Causal measurement in quantum field theory: spacetime}}
\author{Robert Oeckl\footnote{email: robert@matmor.unam.mx}}


\affil{Centro de Ciencias Matemáticas, \\
Universidad Nacional Autónoma de México, \\
C.P.~58190, Morelia, Michoacán, Mexico}

\date{UNAM-CCM-2025-2\\ 9 November 2025\\ 2 April 2026 (v2)}

\maketitle

\vspace{\stretch{1}}

\begin{abstract}
\input{abstract}
\end{abstract}

\vspace{\stretch{1}}
\end{titlepage}

%% file: abstract.tex
We provide a framework and explicit construction for the regularized measurement of a large class of spacetime-localized observables in bosonic quantum field theory. The measurements fully satisfy relativistic causality and causal transparency, i.e., avoid unphysical superluminal signaling. We show explicitly how the measurement of time-extended observables back-reacts on itself and induces correlations between other measurements in its causal future. Our framework is fully compositional in spacetime and extends previous results on the measurement of instantaneous observables.

%% file: intro.tex

\section{Introduction}
\label{sec:intro}

A notion of measurement has been an integral part of the standard formulation of quantum theory from the very beginning. This notion is non-relativistic in the sense that measurements are instantaneous and a priori not localized in space. It is only with quantum field theory (QFT) that space and special relativity become a manifest part of the formalism of quantum theory. However, the development of a notion of quantum measurement fully adapted to this context has faced important obstacles.

A first issue is the construction of measurements that are localized in space as well as in time. Superficially, the answer seems obvious, as point-localized field operators of the form $\hat{\phi}(t,x)$ have been part of QFT from the beginning. An operator is indeed useful for describing a single measurement, or even a joint measurement, where all points are relatively spacelike localized \cite{ArHa:collisionlocalobs}. However, we need more information if we want to consider a general sequence of localized measurements. In particular, we would need to know the eigenvalues and eigenspaces of the operator, i.e., its spectral decomposition. This is problematic due to the highly singular nature of point-localized field operators. On the one hand, a smearing in space (or spacetime) is necessary to obtain a well-defined operator on Hilbert space. On the other hand, even then, due to the unbounded nature of this operator, the spectral decomposition is described not by ordinary eigenspaces, but by a \emph{spectral measure}.
Surprisingly, explicit expressions for the spectral measure of field operators have been worked out only recently \cite{Oe:spectral}.

Rather than measuring a field operator itself one may attempt to use it to construct a well-defined quantum operation in a more indirect way. A particularly relevant case is that of a Gaussian quantum operation of the form
\begin{equation}
   \sigma\mapsto
   \exp\left(-\alpha\left(\hat{\phi}(t,x)-q\right)^2\right)
   \sigma
   \exp\left(-\alpha\left(\hat{\phi}(t,x)-q\right)^2\right) .
   \label{eq:gop}
\end{equation}
This can be interpreted as encoding an approximate measurement of whether the field $\phi$ at $(t,x)$ has the value $q$, with $\alpha$ controlling the approximation. With a non-relativistic position operator instead of a field operator such Gaussian operations have been considered as approximate measurements of particle positions since the 1980s \cite{BaLaPr:marcocontobsqm,GhRiWe:micromacro}. The use of a field operator in this context has been suggested in \cite{BaLaPr:opvalstochastic}.

A second issue is the requirement that measurements should respect the causal structure of spacetime in not permitting superluminal signaling. We refer to this property as \emph{causal transparency}. Since the measurement process in QFT is inherited from non-relativistic quantum mechanics, it does not inherently conform to special relativity. In particular, Sorkin has shown in a seminal paper in 1993 that (a large class of) projective measurements violate causal transparency and are thus unphysical \cite{Sor:impossible}. Since projective measurements are ubiquitous and include those arising from the  spectral decomposition of observables, this has been seen as a serious impediment to a satisfactory theory of measurement in QFT. For a more refined analysis, see \cite{BoJuKe:impossible}. Recently, Jubb has shown that non-selective quantum operations corresponding to those of the type \eqref{eq:gop} (via suitable integration over $q$) satisfy causal transparency \cite{Jub:causalupdates}. Even more recently, explicit expressions for the spectral measure of field operators have been obtained \cite{Oe:spectral}. A regularization of the spectral measure leads to quantum operations generalizing \eqref{eq:gop}, allowing to measure field operators without inducing superluminal signaling. With this, a multi-observable quantum-operation-valued functional calculus is constructed, permitting the regularized measurement a large class of observables generated by field observables.

A third issue are the limitations of the operator formalism for the description of joint measurements in spacetime. A priori this allows only instantaneous measurements arranged on consecutive spacelike hypersurfaces. On the one hand we need to be able to describe measurements that are not only extended in space, but also in time. Progress has been made in the particular case of continuous measurements. We merely note here that the path integral has been identified as a suitable framework to describe these \cite{BaLaPr:marcocontobsqm}. On the other hand we want to consider joint measurements that are freely arranged in space and time. This requires a conceptually more flexible framework for quantum theory than the standard formulation. We shall adopt the \emph{positive formalism} to this end \cite{Oe:dmf,Oe:posfound}. The central object there is the notion of \emph{probe} that generalizes the notion of \emph{quantum operation} from the operator setting. Together with path integral quantization and the \emph{Schwinger-Keldysh formalism} \cite{Kel:diagramne} this provides a suitable toolbox for a fully spacetime compositional description of quantum measurement in QFT \cite{OeZa:lcmeasure}.

The aim of this paper is to provide a fully compositional construction of probes that encode the regularized joint measurement of a large class of spacetime extended observables. Technically this is achieved by promoting the operators and quantum operations arising from the regularized spectral decomposition of field operators from \cite{Oe:spectral} to probes in the spacetime setting of \cite{OeZa:lcmeasure}. The main result is that these probes are not only spacetime localized, but satisfy causal transparency (Theorem~\ref{thm:ctst}), generalizing the results of \cite{Oe:spectral}.

We emphasize that the scope of this paper is strictly limited to measurement at the fundamental level of the quantum formalism. In particular, we do not consider scenarios where 
the measurement apparatus is explicitly modeled as a separate ancilla system, which interacts unitarily with the system to be measured \cite{vne:mathgrundquant}. The outcome is then read out at a later time by a proper measurement on the ancilla system. There is a considerable literature on this approach to describing measurement in QFT. It can be roughly divided into two directions depending on whether the ancilla system is itself a field theory \cite{HeKr:opmeasureii,FeVe:qftlocalmeasure} or is a non-relativistic system, often modeled as an Unruh-DeWitt detector \cite{BiDa:qftcurved,PGGaMM:detectormeasurementqft}. A good recent review focused mostly on these approaches is \cite{PaFr:eliminatingimpossible}. Relating the present work to these approaches would be important and timely, but is out of scope here.

This paper is about bosonic quantum field theory. In some paragraphs, when this is clearly indicated, the more restricted setting of a real scalar field is considered. This restriction is useful due to its simplicity, while still capturing most of the relevant issues of locality, causality and compositionality. However, all results obtained and the formalisms they are expressed in apply to general bosonic fields.

An outline of this paper is as follows: We recall in Section~\ref{sec:loccaus} the problem of causality in non-relativistic and relativistic quantum measurement, with a particular emphasis on Sorkin's findings \cite{Sor:impossible}. In Section~\ref{sec:nrobs} we recall the standard description of the quantum measurement process based on quantum operations and observables, first in non-relativistic quantum mechanics and then in QFT, including in the latter case the latest results on causal transparency in the regularized measurement of observables \cite{Oe:spectral}. The main section of the paper is Section~\ref{sec:spacetime}. After assembling the relevant ingredients including amplitudes, correlation functions, propagators, spacetime and slice observables, quantum operations and probes in the framework of \cite{OeZa:lcmeasure}, the probes encoding regularized measurements of spacetime observables are constructed. Crucially, it is demonstrated that these satisfy causal transparency. Composite observables are addressed in Section~\ref{sec:multiobs} with explicit expressions for their measurement. A preliminary analysis of causal correlations in quantum measurement in QFT, based on the preceding results is provided in Section~\ref{sec:causalcor}. A discussion of results and an outlook is presented in Section~\ref{sec:outlook}. Appendices~\ref{sec:mathspec}, \ref{sec:mathst}, and \ref{sec:pmultiobs} contain some of the explicit calculations and proofs of statements in Sections~\ref{sec:instqft}, \ref{sec:spacetime}, and \ref{sec:multiobs} respectively.

%% file: causaltransp.tex

\section{Locality and causality of measurement}
\label{sec:loccaus}

We recall elementary notions of locality and causality for measurement in non-relativistic quantum mechanics and in QFT. For the moment we make the usual assumption that measurements are instantaneous.

\subsection{Causality in non-relativistic quantum mechanics}
\label{sec:nrcaus}

In non-relativistic quantum mechanics we require measurements to satisfy \emph{causality}, because this is a feature of nature we observe in all actual measurements. It means that the future choice of performing one measurement over another cannot have any bearing on the probabilities of outcomes of a measurement performed at present. In the standard formulation of quantum theory this is enforced by the \emph{causality axiom}: A \emph{non-selective quantum operation} must preserve the trace of the state. Indeed, this axiom is of such transcendence that the property of being trace-preserving is commonly used as a \emph{definition} of non-selectiveness of a quantum operation.

Denoting the quantum operation encoding the present measurement by $M$ and that encoding the future one by $N$ we require the equality
\begin{equation}
  \tr(N\circ M(\sigma))=\tr(M(\sigma))
  \label{eq:nrctwo}
\end{equation}
for any state $\sigma$.\footnote{We use the symbol $\circ$ to denote the composition of operators. This corresponds here to the temporal order from right (past) to left (future).} Here, $N$ is non-selective while $M$ may be selective. As this condition must hold for any $M$ and $\sigma$ it is simply equivalent to the condition on $N$, for any $\sigma$,
\begin{equation}
  \tr(N(\sigma))=\tr(\sigma).
  \label{eq:nrcaus}
\end{equation}

\subsection{Relativistic causality in QFT}
\label{sec:relcaus}

\begin{figure}
  \centering
  \includegraphics[width=0.6\textwidth]{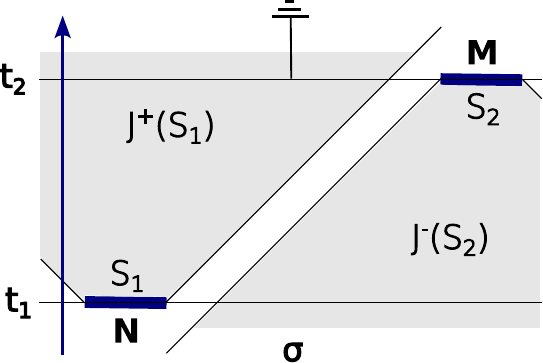}
\caption{Setup with two measurements at different times that are localized in space. The non-selective measurement $N$ is localized in the spatial subset $S_1$ at $t_1$ and the selective measurement $M$ is localized in the spatial subset $S_2$ at $t_2$. $S_1$ does not intersect the causal past of $S_2$. The initial state is $\sigma$ at $t_1$ and the system is discarded after time $t_2$.}
\label{fig:stlocality}
\end{figure}

In QFT we want the same causality axiom to hold of course. However, to reflect the finite speed of propagation of information we require a stronger version. Namely, the probabilities of outcomes of our measurement $M$ in the present should be independent even of a past measurement choice $N$ as long as that past measurement is outside our \emph{causal} past. That is, we require in this situation even $\tr(M\circ N(\sigma))=\tr(M(\sigma))$ for any state $\sigma$. Again, $N$ is non-selective, while $M$ may be selective. We refer to this property as \emph{relativistic causality}. It is intimately related to \emph{locality} as it has to do with the localization of the measurements in spacetime as illustrated in Figure~\ref{fig:stlocality}. We allow for the measurements to take place on subsets $S_1, S_2$ of spacelike hypersurfaces (here the equal-time hypersurfaces at $t_1$ and $t_2$). Crucially, the causal past $J^-(S_2)$ of subset $S_2$ where measurement $M$ takes place and the causal future $J^+(S_1)$ of subset $S_1$ where measurement $N$ takes place, do not intersect.

\subsection{Causal transparency in QFT}

\begin{figure}
  \centering
  \includegraphics[width=0.6\textwidth]{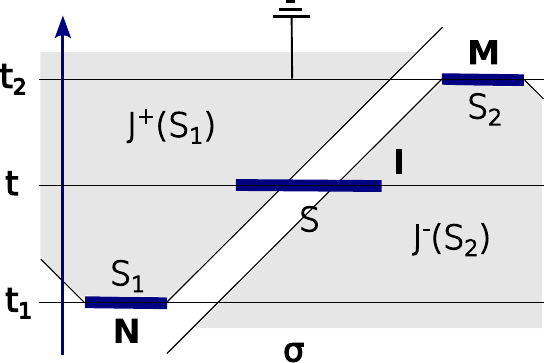}
\caption{An additional non-selective measurement $I$ at the intermediate time $t$ is inserted into the setting of Figure~\ref{fig:stlocality}.}
\label{fig:stcausalt}
\end{figure}

There is another way in which a quantum measurement may, even if it is local in the sense just discussed, disrespect the causal structure of spacetime. Namely, a measurement $I$ localized in a spatial subset $S$ at the intermediate time $t$ may violate causality by facilitating superluminal signaling within $S$. To detect this, Sorkin proposed a setup involving three measurements \cite{Sor:impossible}, see Figure~\ref{fig:stcausalt}. As in the previous setting, there is an initial non-selective measurement $N$. Later, the measurement $I$ takes place, which is also non-selective. Finally, the possibly selective measurement $M$ takes place. As before, the causal future $J^+(S_1)$ of $S_1$, where $N$ takes place and the causal past $J^-(S_2)$ of $S_2$, where $M$ takes place do not intersect. On the other hand, the subset $S$ where $I$ takes place may intersect both $J^+(S_1)$ and $J^-(S_2)$. As before, the measurement $M$ should not be able to detect whether the measurement $N$ takes place. This should be true, even if measurement $I$ takes place as well. Otherwise, we may infer that superluminal signaling has taken place in $S$ due to the measurement $I$. We call the absence of such superluminal signaling for $I$ in all situations of this kind \emph{causal transparency} of $I$, because it means that $I$ is transparent to the underlying causal structure of spacetime \cite{Oe:spectral}. Mathematically this is expressible in the identity $\tr(M\circ I\circ N(\sigma))=\tr(M\circ I(\sigma))$ for any state $\sigma$.

\subsection{Sorkin's result}
\label{sec:sorkincaus}

In a seminal paper, Sorkin has shown that a large class of projective measurements, i.e., measurements encoded by non-selective quantum operations of the form
\begin{equation}
  \sigma\mapsto P\sigma P + (1-P)\sigma (1-P),
\end{equation}
where $P$ is a projector, violate causal transparency, i.e., allow for superluminal signaling in the sense discussed. However, quantum operations based on projectors provide the simplest and most straightforward implementation of measurements with a finite number of alternative outcomes. What is more, the measurement of an arbitrary observable given in the form of a self-adjoint operator is encoded through quantum operations constructed from the projectors of the spectral decomposition of the operator as we recall in Section~\ref{sec:nrquant}. For these reasons Sorkin's result has been viewed as a kind of no-go theorem in quantum measurement theory.
However, recently it has been shown that (regularized) measurements for large classes of observables can indeed be implemented in QFT in a causally transparent fashion at a fundamental level \cite{Oe:spectral}. Key ingredients were the construction of the measurements as limits of non-projective quantum operations and the measurement of a continuous real spectrum of outcomes instead of a discretization into bins, see Section~\ref{sec:instqft}.

%% file: quantization.tex

\section{Observables, quantization, measurement}
\label{sec:nrobs}

\subsection{Non-relativistic quantum mechanics}
\label{sec:nrquant}

\emph{Quantization} usually refers to a procedure to associate to a \emph{classical observable} $A$, i.e., a real valued function on the (instantaneous) phase space $L$, a self-adjoint operator $\hat{A}$ on the Hilbert space $\cH$ of pure states of the theory.\footnote{This operator is often also called a (quantum) observable in the literature. However, for clarity, we reserve the word \emph{observable} for classical observables.} We may use this to describe the \emph{measurement} of different quantities associated with the observable. In the simplest case we are only interested in the \emph{expectation value} of the observable and \emph{discard} the system after the measurement, i.e., choose to ignore its future evolution. Denote by $\cB$ the partially ordered real vector space of self-adjoint operators on $\cH$. Given an initial normalized state $\sigma\in\cB$ (and assuming $\hat{A}$ to be bounded), this expectation value is given by
\begin{equation}
  \langle A\rangle_\sigma=\tr(\hat{A}\sigma) .
  \label{eq:evsingled}
\end{equation}
If we are interested in measuring a different quantity associated to the observable or wish to involve information about the future of the system after measurement, we need to construct the corresponding \emph{quantum operation} from the operator $\hat{A}$. This involves first obtaining the \emph{spectral decomposition} of $\hat{A}$. Recall that the \emph{spectrum} $S_A\subseteq\R$ is the set of values $a$ such that $\hat{A}-a\id$ is not invertible. Here, $\id$ denotes the identity operator. In the finite-dimensional case the spectrum is finite and consists of \emph{eigenvalues}. Then, we obtain a decomposition of $\hat{A}$ into a linear combination of orthogonal projectors $P_a$ with the eigenvalues as coefficients,
\begin{equation}
  \hat{A}=\sum_{a\in S_A} a P_a, \qquad P_a P_b=0\quad\text{for}\quad a\neq b,
  \qquad \sum_{a\in S_A} P_a=\id.
  \label{eq:finspecdec}
\end{equation}
The elements $a$ of the spectrum are the possible \emph{outcome values} of the measurement. The quantum operation that allows to extract the probability of outcome $a$ is the \emph{selective} operation
\begin{equation}
  \cA_a(\sigma)= P_a\sigma P_a .
  \label{eq:selqproj}
\end{equation}
The corresponding \emph{non-selective} operation is the sum over the selective ones,
\begin{equation}
  \cA_*(\sigma)= \sum_{a\in S_A} P_a \sigma P_a ,
\end{equation}
satisfying the ordinary causality axiom \eqref{eq:nrcaus}, i.e., $\tr(\cA_*(\sigma))=\tr(\sigma)$. Crucially, apart from recovering the outcome probability $\Pi_a=\tr(\cA_a(\sigma))$ or the expectation value $\langle A\rangle_{\sigma}=\tr(\sum_{a\in S} a \cA_a(\sigma))$ for a single measurement with discard, the quantum operations serve to extract probabilities and correlations of joint measurements. What is more, the description of joint measurements is compositional in time. That is, a joint measurement is encoded in the composition (as superoperators) of the quantum operations corresponding to the individual measurements, in their temporal order. For example, measuring the observable $B$ after the observable $A$ with initial state $\sigma$ and then discarding the system, the joint probability for outcome $a$ for $A$ and $b$ for $B$ is $\tr(\mathcal{B}_b(\cA_a(\sigma)))$.

Given a self-adjoint operator $\hat{A}$, there are further measurements than those indicated so far, that we may associate with it. In particular, we may \emph{coarse-grain} the measurement. Let $X\subseteq S_A$ be a subset of the spectrum and define the associated projector $P_X\defeq \sum_{a\in X} P_a$. Then, the quantum operation
\begin{equation}
  \cA_X(\sigma)\defeq P_X \sigma P_X
  \label{eq:selqsubp}
\end{equation}
encodes a measurement to determine whether or not the value of the observable $A$ lies in the subset $X$. With the previously considered quantum operations \eqref{eq:selqproj} we may pose the same question via the quantum operation
\begin{equation}
  \cA_X'(\sigma)\defeq \sum_{a\in X} \cA_a(\sigma)= \sum_{a\in X} P_a\sigma P_a .
\end{equation}
While the question is the same, the measurements encoded by $\cA_X$ and $\cA_X'$ are not, and the probabilities they yield are different in general. With $\cA_X'$ the full information as to whether any one of the particular values in $X$ is obtained as the outcome is extracted from the quantum system, even though only part of that information is used. In contrast, with $\cA_X$ the information extracted from the quantum system only concerns the yes/no question whether the outcome pertains to $X$ or not. To define a complete coarse grained measurement we have to specify the alternatives. That is, we need a partition $\check{X}=\{X_k\}_{k\in I}$ of the spectrum $S$ in terms of disjoint non-empty subsets $X_k$. Then, the non-selective quantum operation for this coarse-grained measurement is
\begin{equation}
  \cA_{\check{X}}(\sigma)= \sum_{k\in I} \cA_{X_k}(\sigma) = \sum_{k\in I} P_{X_k} \sigma P_{X_k} .
  \label{eq:nsqpart}
\end{equation}
Again, this satisfies the causality axiom \eqref{eq:nrcaus}.
Crucially, this quantum operation depends on the partition. Note that the original and maximally fine-grained notion of measurement is recovered as the special case where the sets $X_k$ consist of one element each.

In the infinite-dimensional case, the spectrum of a self-adjoint operator may be infinite and even continuous. In the general case, where the operator might even be unbounded, the sum in \eqref{eq:finspecdec} is replaced by an integral with respect to the \emph{spectral measure} $\mu_A$. We may write this as,
\begin{equation}
  \hat{A}=\int_S \lambda\, \xd\mu_A(\lambda) .
\end{equation}
This equation has to be interpreted with care (and in the weak operator topology). A more tangible characterization is the following. For any Borel subset $X\subseteq S$, there is a projection operator $P_X\defeq \mu_A(X)$. Moreover, this assignment is countably additive and $P_S=\mu_A(S)=\id$. We can construct associated quantum operations exactly as in the case of coarse-grained measurements. That is, we start with a partition $\check{X}=\{X_k\}_{k\in I}$ of the spectrum $S$. In contrast to the finite-dimensional case the partition may be infinite, but needs to be countable. Moreover, the sets $X_k$ need to be Borel measurable. We then have selective quantum operations $\cA_X$ given by expression \eqref{eq:selqsubp} and a corresponding non-selective quantum operation $\cA_{\check{X}}$ given by expression \eqref{eq:nsqpart}, with the exact same interpretations. A key difference to the finite-dimensional part arises if the spectrum is (possibly partially) \emph{continuous}.
This case is important, as for example the position and momentum operators of a free particle have continuous spectrum $S=\R$. In this case a \emph{maximally fine-grained} measurement as a special case of the coarse-grained ones \emph{does not exist}, as a partition of the spectrum consisting of one-element sets would be uncountable. Consequently, any measurement constructed in this way has to be genuinely coarse-grained. In particular, for any choice of partition, there will be uncountably many distinct outcome values that cannot be distinguished by the measurement.

From an operational point of view, the central object of the quantum theory that corresponds to the classical observable $A$ is not the operator $\hat{A}$, but the quantum operations $\cA_a$, $\cA_*$ etc.\ used to encode its measurement in the general case. We take this as justification to refer to the procedure of constructing these quantum operations from the classical observable $A$ as \emph{operational quantization}. Of course the present scheme "passes through" the operator $\hat{A}$, obtained by quantization as understood traditionally, but that need not be the case for other schemes as we will see.

It is important to emphasize that a priori, the measurement of a classical observable $A$ on the instantaneous phase space $L$ is an \emph{instantaneous measurement}. This is consequently also true for the quantum measurement described through a quantization of the observable $A$. In the classical theory we may have the liberty to reinterpret the measurement of $A$ as a measurement taking place at a different time or as being extended in time by making use of an identification of $L$ with a space of global solutions of the equations of motion. We do not in general have this liberty in the quantum theory. On the one hand the evolution of the quantum system cannot in general be brought into correspondence to the classical one. On the other hand this is because the measurement alters the quantum system at the time it is performed. (We return to this point later.) This does not mean that we cannot describe time-extended measurements in the quantum theory. Indeed, composing instantaneous measurements with time-evolutions leads to time-extended measurements. One such time-evolution may describe the interaction of a subsystem of interest with an ancilla subsystem representing a measurement apparatus. One may then legitimately ascribe the duration of the measurement not only to the instantaneous read-out of the ancilla, but to the whole composite quantum operation that includes the interaction.

\subsection{QFT}
\label{sec:instqft}

\begin{figure}
  \centering
  \includegraphics[width=0.6\textwidth]{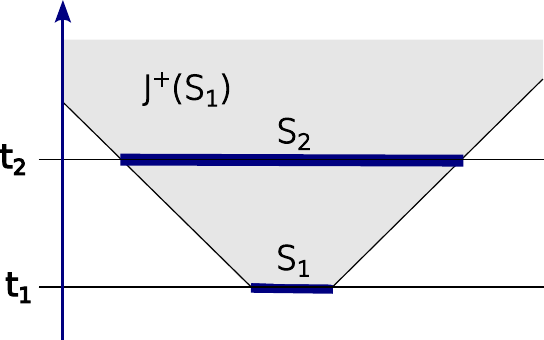}
\caption{If an operator at $t_1$ is localizable in the spatial subset $S_1$, then it is localizable at $t_2$ in $S_2$.}
\label{fig:causaldep}
\end{figure}

The analysis of the non-relativistic case applies also to QFT. What changes is that localization not only in time, but also in space becomes an intrinsic ingredient of the formalism. The instantaneous phase space $L$ is now the space of \emph{initial data} of the equations of motion with explicit spatial localization of field data. With this comes the corresponding \emph{spatial localizability} of classical observables. We say that an observable $A$ is \emph{localizable} in a subset $S$ of space if its support is contained in initial data on $S$. This localizability is inherited by the self-adjoint operators obtained trough (traditional) Weyl \emph{quantization} as follows. If two observables have disjoint support in space on initial data, the corresponding operators commute. What is more, the localization properties extend from space to spacetime. If an operator at time $t_1$ is localizable in a subset $S_1$ of space, then the time-evolved version of this operator at time $t_2$ is localizable in the subset $S_2$ of space which is obtained as the causal future of $(t_1,S_1)$ at time $t_2$, see Figure~\ref{fig:causaldep}.

By the spectral theorem the commutation properties of an operator are inherited by any other operator obtained as an integral over its spectral measure. This applies in particular to the projection operators used in the construction of quantum operations, to implement the measurement of an observable, compare equation \eqref{eq:selqsubp}. As a consequence, the quantum operations also inherit the localization properties of the observables from which they are constructed. In particular, quantum operations constructed from self-adjoint operators which commute due to their localization properties, also commute. On the other hand, we already know that non-selective quantum operations constructed from observables as described satisfy the (non-relativistic) causality axiom (Section~\ref{sec:nrcaus}). Combining both types of properties it is straightforward to show that these operations satisfy \emph{relativistic causality} (Section~\ref{sec:relcaus}).

In contrast, \emph{causal transparency} is a more delicate property to satisfy as we have recalled in Section~\ref{sec:sorkincaus}. In particular, all quantum operations we have considered so far are based on projection operators as in expression \eqref{eq:selqsubp}. Sorkin's result suggests that these will lead to superluminal signaling, violating causal transparency, although his argument is based on a more specific class of projection operators. To get a better handle on observables in QFT, their quantization, and the question on causal transparency we start with \emph{linear observables} $A:L\to\R$. The self-adjoint operators $\hat{A}$ obtained from them through (traditional) Weyl \emph{quantization} are the \emph{field operators}. (We assume $L$ to be a linear space.) The latter can alternatively be characterized as those operators that can be written as the sum of a creation and corresponding annihilation operator. 

In the example of a real scalar field, the quintessential point-localized observable is $(\phi,\dot{\phi})_t\mapsto\phi(x)$ with the associated self-adjoint operator denoted by $\hat{\phi}(t,x)$. Here $t$ is the time at which the initial data $L$ is taken and $x$ is the spatial position where the field data is evaluated. Since initial data consists of field values and temporal derivatives, general linear observables can be parametrized by two real (and suitably square-integrable) functions $g,g':\R^3\to\R$ on space via
\begin{equation}
  (\phi,\dot{\phi})_t\mapsto \int_{\R^3}\xd^3 x\, \left(g(x)\phi(x) + g'(x)\dot{\phi}(x)\right) .
  \label{eq:classlinobs}
\end{equation}
We denote the associated self-adjoint operator by $\hat{\phi}(t,g,g')$. The observable and operator are then \emph{localizable} on the hypersurface at time $t$ in the subset $S$ of space that is the union of the supports of $g$ and $g'$. In spacetime, their time-evolved versions are localizable in the region given by the union of the causal past and future of $(t,S)$, intersected with the corresponding spacelike hypersurface to which they are evolved. When no confusion can arise, we simply say that they are localizable in the corresponding spacetime region, without mentioning the hypersurface. In the same vein, it is customary to parametrize observables and corresponding operators by functions $f:\R\times\R^3\to\R$ on spacetime. The corresponding operator is given by
\begin{equation}
  \hat{\phi}(f)=\int_{\R^4}\xd t\, \xd^3 x\, f(t,x) \hat{\phi}(t,x).
\label{eq:stop}
\end{equation}
It is then easy to see that in our previously established manner of speaking, this operator is localizable in the spacetime region that is the union of the causal past and future of the support of the function $f$. By identifying the spaces of initial data with the space of global solutions this operator can be thought of as arising from the quantization of the observable
\begin{equation}
   \phi\mapsto \int_{\R^4}\xd t\, \xd^3 x\, f(t,x) \phi(t,x)
\label{eq:stobs}
\end{equation}
on the \emph{space of solutions}.

We proceed to consider the measurement of linear observables and the corresponding field operators. As a first remark, it is important to recall that even if we parametrize observables and operators in an apparently temporally extended way, compare expression \eqref{eq:stop}, we have not left the standard quantization framework where measurements are \emph{instantaneous}. In particular, if the measurement of an observable determined by a spacetime function $f$ is applied at time $t$, there are unique functions $g,g'$ on the instantaneous phase space at $t$ that encode this observable, compare expression \eqref{eq:classlinobs}. This determines the a priori physical interpretation of the measurement.

A linear observable $A$ and by consequence the self-adjoint operator $\hat{A}$ obtained from it by quantization have the real numbers as their spectrum, $S=\R$. The construction of quantum operations encoding a measurement of the observable may proceed as outlined in Section~\ref{sec:nrquant}. Thus, using projection operators $P_X$ obtained from the spectral measure on a subset $X\subseteq S=\R$, we may perform coarse-grained measurements via quantum operations $\cA_X$, compare expression \eqref{eq:selqsubp}. As previously mentioned, Sorkin's result, although based on a different class of projection operators, suggests that this measurement would not satisfy causal transparency. This violation of causal transparency was confirmed by Albertini and Jubb precisely for the projectors obtained from linear observables for the real scalar field \cite{AlJu:measurecausal}.

Clearly, a different approach is needed to satisfy causal transparency. The key seems to be that in spite of the spectrum being uncountable, we should insist on a maximally fine-grained measurement. The apparent obstacle is that the spectral measure does not assign projectors to points, which would enable us to define a quantum operation as in expression \eqref{eq:selqproj}. However, it \emph{almost} does. More precisely, we can construct a one-parameter-family of \emph{positive-operator-valued measures (POVMs)} that converge to the spectral measure in a suitable sense and which has the desired properties. Let $A:L\to\R$ be a linear observable. For $\epsilon>0$ and $q\in\R$ we define the induced observable
\begin{equation}
  H_A^\epsilon(q)(\phi)\defeq \frac{1}{\sqrt{\pi}\epsilon}\exp\left(-\frac{1}{\epsilon^2} (A(\phi)-q)^2\right) .
  \label{eq:aeobs}
\end{equation}
We write $\Pi_A^\epsilon(q)\defeq \hat{H}_A^\epsilon(q)$ for the self-adjoint operator obtained from $H_A^\epsilon(q)$ by Weyl quantization. For any $\epsilon>0$ and $q\in\R$, $\Pi_A^\epsilon(q)$ is bounded and positive and has the same localization properties as $A$ and $\hat{A}$. In particular, if $\hat{A}$ commutes with $\hat{B}$, then so does $\Pi_A^\epsilon(q)$.

For any $A$ and $\epsilon$, the family $\Pi_A^\epsilon(q)$ gives rise to a POVM as follows. Let $f:\R\to\R$ be Lebesgue measurable and essentially bounded. Then,
\begin{equation}
  \Pi_A^\epsilon[f]\defeq \int_{-\infty}^\infty \xd q\, f(q) \Pi_A^\epsilon(q)
  \label{eq:rfcalc}
\end{equation}
is a bounded operator, and it is positive if $f$ is positive, i.e., if $f\ge 0$. Also, $\Pi_A^\epsilon[\one]=\id$, where $\one(q)=1$ denotes the constant function with value $1$ and $\id$ is the identity operator. Moreover, $\Pi_A^\epsilon[f]$ also inherits the localization properties of $A$ and $\hat{A}$.
The analogs of the previously considered projectors for subsets $X\subseteq S=\R$ of the spectrum are the operators $\Pi_A^\epsilon[\chi_X]$, where $\chi_X$ denotes the characteristic function of the subset $X$. Crucially, for $\epsilon\to 0$ the operator $\Pi_A^\epsilon[f]$ converges (strongly) to $\Pi_A[f]\defeq\Pi_A^0[f]$, which coincides with the integral over $f$ with the spectral measure associated to $\hat{A}$. (Weak convergence was shown in \cite{Oe:spectral} and a proof of strong convergence was subsequently given in \cite{MaNa:qftfv}.) Indeed, another useful way to understand $\Pi_A^\epsilon[f]$ is as arising from the spectral measure integrated with a Gaussian convolution $f^{\epsilon}$ of $f$, i.e., $\Pi_A^\epsilon[f]=\Pi_A[f^{\epsilon}]$, see Appendix~\ref{sec:mathspec} for details. 
$\Pi_A[\chi_X]$ is precisely the previously considered projection operator $P_X$. However, while $\Pi_A^\epsilon(q)$ is continuous in $\epsilon$ for $\epsilon>0$, a limit $\epsilon\to 0$ does not exist.

For $\epsilon>0$ define the quantum operation
\begin{equation}
  \cA^\epsilon(q)(\sigma)\defeq \sqrt{2\pi}\epsilon\, \Pi_A^\epsilon(q) \sigma \Pi_A^\epsilon(q) .
  \label{eq:rsmop}
\end{equation}
This can be considered as encoding a measurement of $A$ in terms of the probability (density) that the outcome coincides with the value $q$, in analogy to expression \eqref{eq:selqproj}. More generally, for $f$ measurable and essentially bounded we define 
\begin{equation}
  \cA^\epsilon[f](\sigma)\defeq \sqrt{2\pi}\epsilon\int_{-\infty}^{\infty}\xd q\, f(q) \Pi_A^\epsilon(q) \sigma \Pi_A^\epsilon(q) .
  \label{eq:fcalcqop}
\end{equation}
This is completely positive if $f$ is positive and thus a quantum operation. If we take a subset $X\subseteq S=\R$ of the spectrum, the quantum operation $\cA^\epsilon[\chi_X]$ encodes the measurement yielding the probability whether the outcome of $A$ lies in $X$ or not. What is more, the corresponding \emph{non-selective} quantum operation, $\cA^\epsilon[\one]$, satisfies the \emph{causality axiom} \eqref{eq:nrcaus}, i.e., $\tr(\cA^\epsilon[\one](\sigma))=\tr(\sigma)$. Together with the fact that these quantum operations inherit the localization properties of the observable $A$ through the operators $\Pi_A^\epsilon(q)$, this implies that \emph{relativistic causality} is also satisfied, as discussed previously.
What is more, the non-selective quantum operations $\cA^\epsilon[\one]$ satisfy \emph{causal transparency}, compare Figure~\ref{fig:stcausalt} and the following theorem.

\begin{thm}[{\cite[Theorem~5.2]{Oe:spectral}}]
  \label{thm:ctrans}
  Let $S_1,S,S_2$ be subsets of the equal-time hypersurfaces at distinct times $t_1,t, t_2$ such that $S_2$ does not intersect the causal future of $S_1$. Let $N$ be a non-selective quantum operation localizable at $S_1$, $M$ a selective quantum operation localizable at $S_2$ and $A:L_t\to\R$ a linear observable with support in $S$. Let $\epsilon>0$ and $I=\cA^{\epsilon}[\one]$ and $\sigma\in\cT$. Then,
  \begin{equation}
     \tr\left((M\comp I\comp N)(\sigma)\right)=\tr\left((M\comp I)(\sigma)\right) .
     \label{eq:ctid}
  \end{equation}
  Here, the symbol $\comp$ means that operations are ordered according to their temporal order and then composed as usual.
\end{thm}

The price we have paid for causal transparency is that we have approximated the spectral measure by a one-parameter-family of POVMs. As already mentioned, the limit $\epsilon\to 0$ does exist for the operators $\Pi_A^\epsilon[f]$, but not for the operators $\Pi_A^\epsilon(q)$. Neither does it exist for the quantum operations $\cA^\epsilon(q)$ or $\cA^\epsilon[f]$. However, the latter limit does exist when composing with the trace. Moreover, the existence of limits of more general composites is suggested by complete positivity together with uniform (in $\epsilon$) boundedness in the trace norm of $\cA^\epsilon[f]$ \cite{Oe:spectral}. Thus, we should think of the quantum operations $\cA^\epsilon[f]$ as encoding \emph{regularized} measurements of the observable $A$. After composing all component operations (including initial states, a final discard etc.) we obtain the desired outcome probability as a function of the \emph{regulator} $\epsilon$. At this point the regulator is to be removed by taking the limit $\epsilon\to 0$ if possible. This yields the actual outcome probability.

Let $X\subseteq S=\R$ be a subset of the spectrum. While responding to the same question as to whether the observable $A$ takes a value in $X$ or not, the measurements in terms of the quantum operation $\cA_X$ given by \eqref{eq:selqsubp} and the quantum operation $\cA^\epsilon[\chi_X]$ are fundamentally different (apart from the fact that the latter is regularized). In the former case only the information whether the outcome lies in $X$ is extracted from the quantum system. In the latter case, the full information about the value of the observable $A$ is extracted and used to decide if it does or not lie in $X$. In particular, the latter measurement is \emph{maximally fine-grained}, even if the question asked is not. This is also reflected in the fact that the corresponding non-selective operation is unique and does not depend on a choice of coarse-graining as in the former case. This seems to be a key property for achieving casual transparency.
On the other hand, even though the introduction of the regularized fine-grained measurement of observables with continuous spectrum was motivated by the causal transparency property in QFT, such measurements are useful also in non-relativistic quantum mechanics as becomes clear already in \cite{BaLaPr:marcocontobsqm}.

A first question about the present construction would be whether it can be generalized to non-linear observables. Indeed, there is no fundamental obstacle to allowing the observable $A$ in expression \eqref{eq:aeobs} to be non-linear, and most further steps go through with appropriate changes. However, the demonstration of causal transparency. i.e., the proof of Theorem~\ref{thm:ctrans} relies in an essential way on the linearity of $A$. Meanwhile, expression \eqref{eq:fcalcqop} suggests a different and possibly better way to deal with non-linear observables. This expression for $\cA^\epsilon[f]$ suggests an \emph{operational functional calculus} for the observable $A$ \emph{valued in quantum operations} with the interpretation of measuring the \emph{expectation value} of the (possibly non-linear) observable $f(A)$. This interpretation is precisely confirmed in the context where we already know from traditional quantization prescriptions how to do this measurement, namely if the system is immediately discarded afterward. Combining Lemma~4.8 of \cite{Oe:spectral} with equation \eqref{eq:evsingled}, we find,
\begin{equation}
  \lim_{\epsilon\to 0} \tr(\cA^\epsilon[f](\sigma))
  =\tr(\Pi_A[f]\sigma)=\tr(\widehat{f(A)}\sigma)=\langle f(A)\rangle_{\sigma} .
\end{equation}
(Note that $\Pi_A[f]$ is just a different notation for $f(\hat{A})$ which in turn is identified in Weyl quantization to $\widehat{f(A)}$.) At the same time this shows how the quantum operation valued functional calculus reduces in this case to the ordinary functional calculus related to the spectral measure of $\hat{A}$. We also recall again that the corresponding non-selective quantum operation is always the same, irrespective of the choice of function $f$, namely $\cA^\epsilon[\one]$, implying that all these measurements satisfy causal transparency. We refer to the procedure of constructing the family of quantum operations $\cA^\epsilon[f]$ from the classical observable $A$ also as \emph{regularized operational quantization}.

Regularized measurements can be composed as through the composition of their quantum operations as maps in temporal order. As usual this yields joint probabilities or expectation values of products. However, the quantum operation valued calculus suggests a more sweeping generalization. Consider linear observables $A_1,\ldots,A_n$, in temporal order, a function $f:\R^n\to\R$, and a regulator $\epsilon>0$. Define the quantum operation
\begin{equation}
  \cA^\epsilon[f]\defeq \int_{R^n} \xd q_1\cdots\xd q_n\, f(q_1,\ldots,q_n)\, \cA_n^\epsilon(q_n)\circ\cdots\circ\cA_1^\epsilon(q_1) .
\end{equation}
This encodes the measurement of the observables in terms of the joint expectation value of $f(A_1,\ldots,A_n)$. Only if the function $f$ factorizes as $f(q_1,\ldots,q_n)=f_1(q_1)\cdots f_n(q_n)$ is this expressible as an ordinary composition of measurements. Nevertheless, the corresponding non-selective measurement is given by the ordinary composite of non-selective quantum operations $\cA_n^\epsilon[\one]\circ\cdots\circ\cA_1^\epsilon[\one]$, ensuring causal transparency. We shall return to this setting in greater generality in Section~\ref{sec:multiobs}.

%% file: spacetime.tex

\section{Spacetime}
\label{sec:spacetime}

In this section we consider observables that are extended not only in space, but also in time, and their measurement. To this end we need a quantization prescription that is capable of dealing with temporally extended observables. The \emph{path integral} is a convenient choice, as it is well established in QFT. An observable in a \emph{spacetime region} $M$ is then given by a map $K_M\to\R$, where $K_M$ denotes the space of field configurations in $M$. We shall denote these observables as \emph{spacetime observables} and the observables defined in Section~\ref{sec:nrobs} as \emph{phase space observables} whenever there is a need to distinguish them explicitly.
We also need a framework to deal with the operational and probabilistic aspects of measurement in a manifestly relativistic setting, generalizing the time-sequential (super-)operator based standard formulation of quantum theory. The \emph{local positive formalism} (PF) is our choice here \cite{Oe:dmf,Oe:posfound}. It plays nicely with QFT by construction, since it originated from \emph{general boundary quantum field theory}, an axiomatic approach to QFT \cite{Oe:holomorphic,Oe:feynobs}.
How the path integral can be used within the PF to describe measurements in QFT was developed in detail recently \cite{OeZa:lcmeasure}, using the \emph{Schwinger-Keldysh formalism}. While we recall the essential setup and some results from that work in the following, we refer the reader to the paper for details. Proofs for statements in this section are collected in Appendix~\ref{sec:mathst}.

\subsection{Amplitudes, observables and correlation functions}
\label{sec:amocor}

We denote the amplitude in the spacetime region $M$ with boundary state $\psi\in\cH_{\partial M}$ by $\rho_M(\psi)$ \cite{Oe:GBQFT}. We can informally define this in terms of the path integral,
\begin{equation}
    \rho_M(\psi)=\int_{K_M}\xD\phi\, \psi(\phi_{\partial M}) e^{\im S(\phi)} .
    \label{eq:ampl}
\end{equation}
Here, $\phi$ is integrated over the space $K_M$ of field configurations in $M$, $\psi(\phi_{\partial M})$ denotes the Schrödinger wave function \cite{Jac:schroedinger,Hat:qft} of $\psi$ evaluated on $\phi_{\partial M}$, the boundary field configuration.
Similarly, we denote the corresponding \emph{correlation function} of the observable $F:K_M\to\R$ by $\rho_M[F](\psi)$. Informally, this is given by the path integral,
\begin{equation}
    \rho_M[F](\psi)=\int_{K_M}\xD\phi\, \psi(\phi_{\partial M}) F(\phi) e^{\im S(\phi)} .
    \label{eq:corfn}
\end{equation}
A translation-invariant measure on the space of field configurations does not really exist. However, the quantities $\rho_M(\psi)$ and $\rho_M[F](\psi)$ can be rigorously defined in the free theory when a suitable complex structure on the boundary phase space $L_{\partial M}$ exists \cite{Oe:holomorphic,Oe:feynobs}.

In order to exhibit explicit formulas, we work with coherent states using the definitions and conventions of \cite{Oe:holomorphic}.
Thus, we denote the complex inner product on the phase space $L$ by $\{\cdot,\cdot\}$. For any element $\phi$ in the phase space there is a \emph{coherent state} $\coh_{\phi}\in\cH$. (Note $\coh_0=\vac$ is the vacuum state.) These have inner product
\begin{equation}
  \langle\coh_{\phi'},\coh_{\phi}\rangle
  =\exp\left(\frac12\{\phi,\phi'\}\right) .
  \label{eq:cohip}
\end{equation}
The coherent states satisfy a \emph{completeness relation} with respect to a Gaussian measure $\nu$,\footnote{The measure $\nu$ really exists on an extension $\hat{L}$ of the phase space $L$ \cite{Oe:holomorphic}, but this detail is unimportant for our present purposes.}
\begin{equation}
  \langle \eta,\psi\rangle = \int_{\hat{L}}\xd\nu(\phi)\,
   \langle \eta,\coh_{\phi}\rangle \langle\coh_{\phi},\psi\rangle .
   \label{eq:cohcompl}
\end{equation}
In the following it is often convenient to use the \emph{normalized coherent states} defined as
\begin{equation}
  \ncoh_{\phi}\defeq\exp\left(-\frac14\{\phi,\phi\}\right)\coh_{\phi} .
\end{equation}
Let $A:K_M\to\R$ be a \emph{linear observable} and $W=\exp(\im A)$ the associated \emph{Weyl observable}. The correlation function of $W$ for a coherent state satisfies the \emph{factorization identity} \cite{Oe:feynobs,CoOe:locgenvac}, 
\begin{equation}
  \rho_M[W](\ncoh_\phi)=\rho_M(\ncoh_\phi) W(\phi^{\Lint}) \rho_M[W](\vac) .
  \label{eq:corfact}
\end{equation}
Here, we are using a direct sum decomposition of the complexified phase space $L_{\partial M}^\C=L_M^\C\oplus L_{X}^{\C}$, written as $\phi=\phi^\Lint+\phi^\Lext$. $L_M^\C\subseteq L_{\partial M}^{\C}$ is the complexified subspace of solutions that continue into the interior of $M$. The expression $\phi^{\Lint}$ shall then also be used to denote this continuation. This makes the second factor in expression \eqref{eq:corfact} well-defined. $L_{X}^\C\subseteq L_{\partial M}^{\C}$ is the subspace determining the \emph{vacuum} in the region $X$ exterior to $M$ \cite{CoOe:vaclag,CoOe:locgenvac}.
The first factor in expression \eqref{eq:corfact} is the amplitude,
\begin{equation}
  \rho_M(\ncoh_\phi)=\exp\left(\im\,\omega_{\partial M}(\phi,\phi^{\Lint})\right) .
\end{equation}
Here, $\omega_{\partial M}$ is the symplectic form on the boundary phase space that is also the imaginary part of the inner product, $\omega_{\partial M}(\phi_1,\phi_2)=\frac12\Im\{\phi_1,\phi_2\}_{\partial M}$.
A remarkable aspect of the second factor in \eqref{eq:corfact} is that it only depends on the value of the observable $W$ on the space of solutions $L_M\subseteq K_M$.  
The third factor is the \emph{vacuum correlation function} of the Weyl observable,
\begin{equation}
  \label{eq:vev}
  \rho_M[W](\vac)=\exp\left(-\frac12 \wuu(A,A)\right),\quad\text{with}\quad
  \wuu(A,A)=\rho[A^2](\vac) . 
\end{equation}
Here, $\wuu$ is the time-ordered 2-observable correlation function or \emph{observable propagator} in analogy to the 2-point function \cite{OeZa:lcmeasure}. In general, the Feynman (or time-ordered) correlation function for two linear observables $A,B$ is defined as
\begin{equation}
  \wuu(A,B)=\rho[A B](\vac),
  \label{eq:feynoprop}
\end{equation}
and it is sensitive to the values the observables take off-shell, i.e., beyond the subspace of solutions.
The utility of the factorization identity \eqref{eq:corfact} lies above all in its use as a generating function in two ways. On the one hand, arbitrary states can be expanded in terms of coherent states and on the other hand Weyl observables generate other types of observables, especially polynomial ones. For a linear observable $A$ we get,
\begin{equation}
  \rho_{M}[A](\ncoh_{\phi})
  = \rho_M(\ncoh_\phi) A(\phi^{\Lint}) .
\end{equation}

The existence of the complex structure on a spacelike hypersurface is guaranteed in particular if canonical quantization succeeds on the hypersurface, e.g., in the sense of Birrel and Davis \cite{BiDa:qftcurved}. We shall further assume global hyperbolicity and that the in- and out-vacua are the same, i.e., there is a global choice of complex structure and corresponding decomposition of solutions into positive- and negative-energy solutions.
Thus, we focus mostly on spacetime regions $M=[\Sigma_1,\Sigma_2]$ delimited by an initial and final spacelike hypersurface, $\Sigma_1$ and $\Sigma_2$ respectively. We also write $\Sigma_1\le\Sigma_2$ to indicate that $\Sigma_1$ is earlier than $\Sigma_2$. The boundary state space decomposes as $\cH_{\partial M}=\cH_{\Sigma_1}\tens\cH_{\overline{\Sigma}_2}$.\footnote{We take $\Sigma_1$ and $\Sigma_2$ to have the same orientation. Since their orientation has to be opposite as boundary components of $\partial M$, we take $\partial M=\Sigma_1 \cup \overline{\Sigma}_2$, where the overline represents inversion of orientation.} The amplitude encodes then the evolution operator $U_{[\Sigma_1, \Sigma_2]}:\cH_{\Sigma_1}\to \cH_{\Sigma_2}$ between initial and final Hilbert space,\footnote{The orientation change for states is encoded in the conjugate-linear map $\iota$.}
\begin{equation}
    \langle \psi_2, U_{[\Sigma_1, \Sigma_2]}\psi_1\rangle_{\Sigma_2}
    =\rho_{[\Sigma_1,\Sigma_2]}(\psi_1\tens\iota(\psi_2)) .
\end{equation}
We shall occasionally employ a similar notation for correlations functions, when convenient,
\begin{equation}
    \langle \psi_2, \widetilde{F}\psi_1\rangle_{\Sigma_2}
    \defeq\rho_{[\Sigma_1,\Sigma_2]}[F](\psi_1\tens\iota(\psi_2)) .
\end{equation}
Here $\widetilde{F}$ is a map $\cH_{\Sigma_1}\to\cH_{\Sigma_2}$ from initial to final Hilbert space. We may also identify the Hilbert spaces via the time-evolution map $U_{[\Sigma_1, \Sigma_2]}$ and in this way consider $\widetilde{F}$ as an operator.

In this situation we can express the decomposition of the phase space in terms of positive- and negative-energy solutions. Per spacelike hypersurface $\Sigma$ we write this as $L_{\Sigma}^\C=L_{\Sigma}^+ \oplus L_{\Sigma}^-$ and $\phi=\phi^+ + \phi^-$. Moreover, by assumption this decomposition is conserved by time-evolution between spacelike hypersurfaces. We shall also implicitly identify the phase space elements with the global space of classical solutions.
Making the unitary time-evolution also implicit, the amplitude can be written as the inner product between coherent states.
With this we may rewrite the factorization identity \eqref{eq:corfact} as,
\begin{equation}
  \rho_{[\Sigma_1,\Sigma_2]}[W](\ncoh_{\phi_1}\tens\ncoh_{\phi_2})
  =\langle \ncoh_{\phi_2},\ncoh_{\phi_1}\rangle W(\phi_1^- + \phi_2^+) \rho_M[W](\vac) .
  \label{eq:ticorfact}
\end{equation}
We note that $\iota(\ncoh_\phi)=\ncoh_\phi$ and $\vac\tens\vac=\vac$. The third factor is the same vacuum correlation function \eqref{eq:vev} that we defined previously. It only depends on the observable and the vacuum, not on the region.

In the case of a scalar theory we can use a spacetime smearing function $f:\R\times\R^3\to\R$ as in expression \eqref{eq:stobs} likewise on the \emph{space of field configurations} instead of the \emph{space of solutions}, defining now a linear spacetime observable $A_f$. Let $W_f=\exp(\im A_f)$ be the corresponding Weyl observable. As already mentioned, the second factor in the factorization identity \eqref{eq:corfact} and \eqref{eq:ticorfact} depends only on $A_f$ on-shell, i.e., on solutions. The vacuum correlation function on the other hand does depend on $A_f$ off-shell, i.e., beyond the space of solutions. It is the exponential of the observable propagator, here,
\begin{equation}
  \wuu(A_f,A_f)=\int_{\R^4\times\R^4} \xd t\,\xd^3 x\,\xd t'\xd^3 x' f(t,x) f(t',x') \wtuu(t,x,t',x'),
  \label{eq:obspsrc}
\end{equation}
where $\wtuu$ is the time-ordered 2-point function or \emph{Feynman propagator}, which we write as,
\begin{equation}
  \wtuu(t,x,t',x')=\langle\vac,\tord\phi(t,x)\phi(t',x')\vac\rangle .
\end{equation}
We use a similar notation for the \emph{Wightman propagator},
\begin{equation}
  \wtdu(t,x,t',x')=\wtud(t',x',t,x)=\langle\vac,\phi(t,x)\phi(t',x')\vac\rangle .
\end{equation}
We recall the different roles that real and imaginary part in these propagators play and how they are related. First, the real parts of both propagators coincide, and we refer to this real and symmetric propagator as the \emph{Hadamard propagator},
\begin{equation}
  \wth(t,x,t',x')\defeq\Re(\wtuu(t,x,t',x'))=\Re(\wtdu(t,x,t',x'))=\Re(\wtud(t,x,t',x')).
\end{equation}
In contrast to the real part, the imaginary parts are \emph{causal} propagators, that is they vanish if the two points are spacelike separated. What is more, in both cases they are linear combinations of the \emph{advanced} and \emph{retarded} propagators, denoted $\wta$ and $\wtr$ respectively,
\begin{align}
  2 \Im(\wtuu(t,x,t',x')) & =\wtr(t,x,t',x')+\wta(t,x,t',x') , \\
  2 \Im(\wtdu(t,x,t',x')) & =\wtr(t,x,t',x')-\wta(t,x,t',x') .
\end{align}
The retarded propagator vanishes if $(t,x)$ is not in the causal future of $(t',x')$ and for the advanced propagator it is the other way round. What is more, they are related as,
\begin{equation}
   \wtr(t,x,t',x')=\wta(t',x',t,x) .
\end{equation}

Crucially, all the propagators have their observable generalizations, and this situation is not limited to scalar field theory. We have seen this for the Feynman (or time-ordered) propagator, compare expression \eqref{eq:feynoprop}. The Wightman observable propagator $\wdu$ can be constructed through a Schwinger-Keldysh double path integral \cite{OeZa:lcmeasure}, as we will recall in Section~\ref{sec:loccomp}, expression \eqref{eq:woprop}. With these two we can reconstruct the other observable propagators. Thus,
\begin{equation}
  \wh(A,B)\defeq\Re(\wuu(A,B))=\Re(\wdu(A,B))=\Re(\wud(A,B)) .
  \label{eq:roprop}
\end{equation}
We also recall that the Hadamard observable propagator is positive in the sense $\wh(A,A)\ge 0$ \cite{OeZa:lcmeasure}. The causal observable propagators are reconstructed from the imaginary parts,
\begin{align}
  \wr(A,B) & =\Im(\wuu(A,B))+\Im(\wdu(A,B)), \\
  \wa(A,B) & =\Im(\wuu(A,B))-\Im(\wdu(A,B)) .
\end{align}
We omit the word "observable" in "observable propagator" in the following when this omission causes no ambiguity.

\subsection{Slice observables}
\label{sec:sobs}

The notion of observable on the phase space $L$ on a (here spacelike) hypersurface $\Sigma$ can be recovered as a degenerate case of the notion of spacetime observable. This is called a \emph{slice observable} \cite{Oe:feynobs,CoOe:locgenvac}.
Roughly, we can imagine a slice observable as a distributional observable concentrated on a hypersurface $\Sigma$. More concretely, adding a linear observable $A$ to the classical action modifies the equations of motion from homogeneous ones to inhomogeneous ones. If the observable is distributional as described, the solutions are homogeneous on each side of the hypersurface $\Sigma$, but jump discontinuously from one side of the hypersurface to the other. There is a one-to-one correspondence between these jumps and elements of the phase space $L_{\Sigma}$ on $\Sigma$ and by duality, with linear observables on the phase space. In this way a linear phase space observable $A'$ determines a linear distributional spacetime observable $A$. In terms of the path integral the modification of the action corresponds to inserting the Weyl observable $\exp(\im A)$. Thereby, we can recover the Weyl quantization of phase space observables through path integral quantization, including their commutation relations.

Conversely, a spacetime observable gives rise to a phase space observable as follows. Assume a spacetime region delimited by two spacelike hypersurfaces $\Sigma_1,\Sigma_2$. The space of solutions $L_{[\Sigma_1,\Sigma_2]}$ is a subspace of the space of field configurations $K_{[\Sigma_1,\Sigma_2]}$, and we can identify the former with the phase space on another spacelike hypersurface $\Sigma$, between $\Sigma_1$ and $\Sigma_2$. In this way a spacetime observable $A$ on $K_{[\Sigma_1,\Sigma_2]}$ gives rise to a phase space observable on $L_{\Sigma}$ and in turn to a slice observable $A_{\Sigma}$ again on $K_{[\Sigma_1,\Sigma_2]}$. It is important to mention that this map from spacetime observables to slice observables is surjective, but not injective.
It is instructive to compare the correlation functions for the two observables. To this end assume $A$ and thus $A_\Sigma$ to be linear observables and $W=\exp(\im A)$ and $W_{\Sigma}=\exp(\im A_{\Sigma})$ to be the corresponding Weyl observables. Both correlation functions, $\rho[W](\ncoh_{\phi_1}\tens\ncoh_{\phi_2})$ and $\rho[W_{\Sigma}](\ncoh_{\phi_1}\tens\ncoh_{\phi_2})$ can be factorized according to expression \eqref{eq:ticorfact}. The first factor coincides trivially since it is the amplitude, which does not depend on any observable. The second factor also coincides by construction, $W(\phi_1^- + \phi_2^+)=W_{\Sigma}(\phi_1^- + \phi_2^+)$, since the observables $W$ and $W_{\Sigma}$ coincide on solutions. The third factor differs due to a difference in the observable propagators $\wuu(A,A)$ and $\wuu(A_{\Sigma},A_{\Sigma})$. However, $\wuu(A_{\Sigma},A_{\Sigma})$ turns out to be not only independent of the choice of spacelike hypersurface, but also related to $\wuu(A,A)$ in a simple way. Namely, the former is the real part of the latter \cite{OeZa:lcmeasure}. With definition \eqref{eq:roprop} we can write,
\begin{equation}
  \wuu(A_{\Sigma},A_{\Sigma})=\wh(A_{\Sigma},A_{\Sigma})=\wh(A,A) .
\end{equation}
Indeed, the imaginary part of the propagator is related to the temporal extension of the observable. Recall that the imaginary part of the 2-point function is the causal part which is non-vanishing only between points that are timelike separated. Then, a slice observable is concentrated on a spacelike hypersurface, where any two points are spacelike separated, hence the imaginary part vanishes. We obtain a relative phase factor,
\begin{equation}
  \rho[W](\ncoh_{\phi_1}\tens\ncoh_{\phi_2})=\rho[W_{\Sigma}](\ncoh_{\phi_1}\tens\ncoh_{\phi_2}) \exp\left(-\frac{\im}{2}\Im(\wuu(A,A))\right) .
  \label{eq:weylphase}
\end{equation}
What is more, for linear observables the quantization is identical,
\begin{equation}
  \rho[A](\ncoh_{\phi_1}\tens\ncoh_{\phi_2})=\rho[A_{\Sigma}](\ncoh_{\phi_1}\tens\ncoh_{\phi_2}) .
\end{equation}
In the operator notation this is, $\widetilde{A}=\widetilde{A}_{\Sigma}$.

We recall the use of a spacetime smearing function $f:\R\times\R^3$ for encoding both phase space observables and spacetime observables in scalar field theory, via the assignment \eqref{eq:stobs}. The phase space observable induced from the spacetime observable $A_f$ on a spacelike hypersurface $\Sigma$ corresponds precisely to the slice observable $A_{f,\Sigma}$ as introduced above. Moreover, the quantization of the latter yields the self-adjoint operator $\hat{\phi}(f)$ as given by \eqref{eq:stop} once we identify the Hilbert spaces $\cH_{\Sigma_1}$ and $\cH_{\Sigma_2}$ of initial and final hypersurface. As noted above, for the linear observables the quantization is identical, in operator notation, $\widetilde{A}_f=\widetilde{A}_{f,\Sigma}=\hat{\phi}(f)$.
For the Weyl observables, however, there is a phase difference. From \eqref{eq:weylphase} and \eqref{eq:obspsrc} we get in operator notation,
\begin{equation}
  \widetilde{W}_f=\exp\left(\im\,\hat{\phi}(f)\right)
  \exp\left(-\frac{\im}{2}\int_{\R^4\times\R^4} \xd t\,\xd^3 x\,\xd t'\xd^3 x' f(t,x) f(t',x') \Im{\wtuu(t,x,t',x')}\right) .
\end{equation}

\subsection{Generalized spectral decomposition}
\label{sec:gspec}

We proceed to transport the notion of (regularized) \emph{spectral decomposition} from the operator setting (Section~\ref{sec:nrobs}) to the spacetime setting. Given a linear observable $A$ on phase space, we recall that the observable $H_A^\epsilon(q)$, compare expression \eqref{eq:aeobs}, defines via Weyl quantization the $\epsilon$-regularized \emph{spectral measure} of the observable $A$. Replacing $A$ with a spacetime observable and Weyl quantization with path integral quantization, the role of the regularized spectral measure is played by the correlation functions of $H_A^\epsilon(q)$ as a spacetime observable. It is straightforward to evaluate it on coherent states.
\begin{prop}
\label{prop:regobs}
\begin{equation}
   \rho_M[H_A^{\epsilon}(q)](\ncoh_{\phi})
   =\rho_M(\ncoh_{\phi})\frac{1}{\sqrt{\pi (\epsilon^2 + 2 \wuu(A,A))}}
      \exp\left(-\frac{(A(\phi^{\Lint})-q)^2}{\epsilon^2+2\wuu(A,A)}\right) .
      \label{eq:spcor}
\end{equation}
\end{prop}
\begin{proof}This is a special case of Proposition~\ref{prop:multicor}, see Section~\ref{sec:multiobs}.
\end{proof}
Recall that the real part of the observable propagator $\wuu(A,A)$ is non-negative, making this well-defined. This correlation function is the analog of the operator $\Pi_A^\epsilon(q)$. Defining the observable
\begin{equation}
  H_A^\epsilon[f](\phi)\defeq \int_{-\infty}^{\infty} f(q) H_A^{\epsilon}(q)(\phi),
\end{equation}
yields upon quantization the correlation function $\rho_M[H_A^\epsilon[f]]$, analog of the operator $\Pi_A^\epsilon[f]$, compare expression \eqref{eq:rfcalc}. Note that here we can take the limit $\epsilon\to 0$ at the level of the observable,
\begin{equation}
f(A(\phi))=H_A[f](\phi)\defeq \lim_{\epsilon\to 0} H_{A}^{\epsilon}[f](\phi) ,
\end{equation}
if $f$ is continuous and increases less than the exponential of a square.

This generalized notion of the regularized spectral measure has many properties of the original spectral measure. In particular, we recover the plain amplitude and the quantized linear observable itself, for any value of $\epsilon$,
\begin{equation}
  \rho_M[H_A^{\epsilon}[\one]]=\rho_M,\quad \rho_M[H_A^{\epsilon}[\id]]=\rho_M[A] .
\end{equation}
We also recover an analog of the spectral functional calculus when taking the limit $\epsilon\to 0$. To see this it is sufficient to recover the Weyl observable $W=\exp(\im D)$ from $w(q)\defeq\exp(\im q)$,
\begin{equation}
  \lim_{\epsilon\to 0} \rho_M[H_A^{\epsilon}[w]](\ncoh_{\phi})=\rho_M[W](\ncoh_{\phi}) .
\end{equation}

The (regularized) spectral decomposition of the operator setting is recovered exactly in the case of slice observables. For time-extended observables, however, the analogs of the (regularized) spectral projectors $\Pi_A^\epsilon[\chi_X]$, compare expression \eqref{eq:rfcalc}, are no longer representable as positive or even self-adjoint operators. This can be read off from the correlation function \eqref{eq:spcor}, which we may write in this case in the operator notation as follows,
\begin{equation}
  \langle\ncoh_{\phi_2}, \widetilde{H_A^\epsilon(q)}\ncoh_{\phi_1}\rangle
  =    \langle\ncoh_{\phi_2}, \ncoh_{\phi_1}\rangle
  \frac{1}{\sqrt{\pi (\epsilon^2 + 2 \wuu(A,A))}}
      \exp\left(-\frac{(A(\phi_1^- +\phi_2^+)-q)^2}{\epsilon^2+2\wuu(A,A)}\right) .
      \label{eq:gspecmatrix}
\end{equation}
Even if the initial and final coherent states are the same, $\phi_1=\phi_2$, this is not necessarily real since the observable propagator $\wuu(A,A)$ is not real in general. If $A$ is a slice observable the imaginary part of $\wuu(A,A)$ does vanish, as previously mentioned.

\subsection{From operations to probes}

The spacetime generalizations of quantum operations are \emph{probes} \cite{Oe:posfound,OeZa:lcmeasure}. The simplest probe is the \emph{null probe}, generalizing unitary evolution of the system in a time-interval. Using the completeness relation \eqref{eq:cohcompl}, we can obtain the null probe $\np_M$ in a region $M$ for a state $\sigma\in\cB_{\partial M}$ in the mixed boundary state space from the amplitude \eqref{eq:ampl},
\begin{equation}
   \np_M(\sigma)=\int_{\hat{L}_{\partial M}}\xd\nu(\phi)
    \rho_M(\sigma \coh_{\phi}) \overline{\rho_M(\coh_{\phi})} .
    \label{eq:defnp}
\end{equation}
The property of being a \emph{completely positive} map $\cB_1\to\cB_2$ of quantum operations finds its generalization for probes in the property of being \emph{positive} maps $\cB_{\partial M}\to\C$. Such probes are called \emph{primitive probes}. Linear combinations of probes are probes and linear combinations of primitive probes with non-negative coefficients are primitive probes. The null probe in particular is a primitive probe. For a region between spacelike hypersurfaces $\Sigma_1,\Sigma_2$ the null probe takes the more familiar form
\begin{equation}
  \np_{[\Sigma_1,\Sigma_2]}(\sigma_1\tens\sigma_2)=\tr(\sigma_2 U_{[\Sigma_1,\Sigma_2]} \sigma_1 U_{[\Sigma_1,\Sigma_2]}^{\dagger}) .
  \label{eq:npunitary}
\end{equation}

As we take the amplitude to be defined by the path integral \eqref{eq:ampl}, we can analogously construct more general probes out of correlation functions, compare expression \eqref{eq:corfn}, through the path integral. Given two observables $F,G:K_M\to\R$ define the probe $\cP_M[F|G]:\cB_{\partial M}\to\C$ for pure states by,
\begin{equation}
  \cP_M[F|G](|\psi\rangle\langle\psi|)=\rho_M[F](\psi) \overline{\rho_M[G](\psi)} .
  \label{eq:probecorpure}
\end{equation}
For genuinely mixed states the completeness relation \eqref{eq:cohcompl} then implies the following formula,
\begin{equation}
  \cP_M[F|G](\sigma)=\int_{\hat{L}_{\partial M}}\xd\nu(\phi)\,
  \rho_M[F](\sigma \coh_{\phi}) \overline{\rho_M[G](\coh_{\phi})} .
  \label{eq:probecor}
\end{equation}
This leads to a version of the Schwinger-Keldysh formalism. We have a double path integral with insertions (here $F$) in the forward-time component and insertions (here $G$) in the backward-time component. Crucially, if $F=G$, then the obtained probe is \emph{primitive}. This was the starting point for the modulus-square construction in \cite{OeZa:lcmeasure}, with the probe taken as encoding a measurement of the expectation value of the observable $|F|^2$.

In case the region is delimited by an initial and final hypersurface, we can represent the probe in a notation resembling a generalized quantum operation,
\begin{equation}
  \cP_{[\Sigma_1,\Sigma_2]}[F|G](\sigma_1\tens\sigma_2)=\tr(\sigma_2 \widetilde{F}\sigma_1\widetilde{G}^{\dagger}) .
\end{equation}
If $F=G$ then $\widetilde{F}$ plays the role of the Kraus operator of a pure quantum operation.

The \emph{discard} as a quantum operation is given by taking the trace. In the spacetime setting we can understand this as a primitive probe associated with the spacetime region to the future of a spacelike hypersurface $\Sigma$. We shall also denote this region by $\Sigma^{\lf}$. Similarly, we denote the region to the past of $\Sigma$ by $\Sigma^{\lp}$. Thus, for $\sigma\in\cB_{\Sigma}$ we write $\discard_{\Sigma^{\lf}}(\sigma)=\tr(\sigma)$.
%

\subsection{Locality and composition}
\label{sec:loccomp}

A key aspect of our approach is \emph{compositionality}, both in the pure state setting of amplitudes and correlation functions, as well as in the mixed state setting of probes. Here and in the following, we use the symbol $\comp$ to denote this composition. In the first case this notion is formalized in \emph{general boundary quantum field theory} through the path integral, in the second it is formalized in the \emph{local positive formalism}, and inherited from the path integral through the Schwinger-Keldysh formalism. We recall this very briefly here, and refer the reader to \cite{OeZa:lcmeasure} for details.

Consider spacetime regions $M_1$ and $M_2$ with disjoint interior, their union given by $M=M_1\cup M_2$. Then, the amplitude for $M$ is the composition of the amplitudes for $M_1$ and $M_2$, symbolically,
\begin{equation}
  \rho_M=\rho_{M_1}\comp\rho_{M_2} .
  \label{eq:amplcomp}
\end{equation}
Let $A_1:K_{M_1}\to\R$ be a spacetime observable in $M_1$. It extends trivially to a spacetime observable in $M$, since we have the forgetful map $K_M\to K_{M_1}$, which we denote with the same name here, yielding the identity,
\begin{equation}
   \rho_M[A_1]=\rho_{M_1}[A_1]\comp \rho_{M_2} .
   \label{eq:amplext}
\end{equation}
Conversely, we say that a spacetime observable $A:K_M\to\R$ is \emph{localizable} in the spacetime region $M_1\subseteq M$ if it can be obtained from a spacetime observable $A_1:M_1\to\R$ by extension in this way.
Similarly, taking an observable $A_2$ in $M_2$ we can make sense of the identity
\begin{equation}
    \rho_{M_1}[A_1]\comp \rho_{M_2}[A_2]=\rho_M[A_1\cdot A_2]
    \label{eq:corrcomp}
\end{equation}
by extending both $A_1$ and $A_2$ to $M$.
For linear observables, we also formalize locality in terms of \emph{additive decomposability}. That is, given a linear observable $K_M\to \R$ with $M=M_1\cup M_2$ as above, we require the existence of linear observables $A_1:K_{M_1}\to\R$ and $A_2:K_{M_2}\to\R$ such that $A=A_1+A_2$ in the obvious way.

The analog of the amplitude in the local positive formalism is the \emph{null probe}, compare expression \eqref{eq:defnp}. With spacetime regions $M=M_1\cup M_2$ as before we have the identity
\begin{equation}
  \np_M=\np_{M_1}\comp\np_{M_2} ,
\end{equation}
inherited from the corresponding relation for amplitudes \eqref{eq:amplcomp}. Similarly, for probes \eqref{eq:probecor} constructed from correlation functions, with observables $F_1,G_1$ localized in $M_1$ we have the identity,
\begin{equation}
   \cP_M[F_1|G_1]=\cP_{M_1}[F_1|G_1]\comp\np_{M_2},
   \label{eq:cprobeext}
\end{equation}
inherited from the identity \eqref{eq:amplext}.
Furthermore, with additional observables $F_2,G_2$ localized in $M_2$ we have, in analogy to \eqref{eq:corrcomp},
\begin{equation}
\cP_{M_1}[F_1|G_1]\comp\cP_{M_2}[F_2|G_2]=\cP_M[F_1\cdot F_2|G_1\cdot G_2] .
\end{equation}
If observables $F$ and $G$ are localizable in regions $K$ and $N$, then they are jointly localizable in the union $K\cup N$ and for a region $M$ with $K\subseteq M$ and $N\subseteq M$ we can define the probe $\cP_{M}[F|G]$. Conversely, we can use an identity modeled on \eqref{eq:cprobeext} to give a general characterization of \emph{localizability} of a probe, beyond the specific construction of probes via correlation functions used here. That is, given a probe $\cM_M$ associated to a spacetime region $M$, and suppose $M=M_1\cup M_2$ is a decomposition with disjoint interiors, we say that $\cM_M$ is localizable in the subregion $M_1\subseteq M$ if there exists a probe $\cM_{M_1}$ (we use intentionally the same symbol $\cM$) associated to $M_1$, such that,
\begin{equation}
   \cM_M=\cM_{M_1}\comp\np_{M_2}.
   \label{eq:probeext}
\end{equation}
As a consequence, we can reassociate probes originally associated to a given spacetime region to another one as long as the new region contains a subregion where the probe is localizable. We will freely use this possibility in the following.

Of course, the considerations concerning composition discussed in this section apply in particular to spacetime regions $[\Sigma_1,\Sigma_2]$ that are delimited by an initial and a final spacelike hypersurface. Indeed, as already mentioned, we will limit ourselves for the most part in this paper to spacetime regions of this type, where we can ensure that objects such as amplitudes, correlation functions and probes, exist and are well-defined. In this context, we note that the universal property of the discard probe to be associated to any spacetime region that is to the future of any spacelike hypersurface can be expressed in a manner analogous to the identity \eqref{eq:probeext}. Say $\Sigma_1\le\Sigma_2$ are spacelike hypersurfaces. Then,
\begin{equation}
   \discard_{\Sigma_1^\lf}=\discard_{\Sigma_2^\lf}\comp\np_{[\Sigma_1,\Sigma_2]}= \discard_{\Sigma_2^\lf}\circ\np_{[\Sigma_1,\Sigma_2]} .
   \label{eq:discext}
\end{equation}
Note that the second equality involves a different symbol for composition:
Here and in the following we denote a spacetime decomposition also by the symbol $\circ$ in the special case that it is along a spacelike hypersurface and in the order of past to future from right to left. This is to emphasize the analogy (and sometimes strict equivalence) in this case to the non-relativistic operator setting (Section~\ref{sec:nrobs}).

Using this notation, we can also provide now the proper definition of the Wightman observable propagator, mentioned already in Section~\ref{sec:amocor}. For linear observables $A,B$ defined in a spacetime region delimited by initial and final spacelike hypersurfaces $\Sigma_1,\Sigma_2$ we have,
\begin{equation}
  \wud(A,B)=\wdu(B,A)=\discard_{\Sigma_2^{\lf}}\circ\cP[A|B]_{[\Sigma_1,\Sigma_2]}(\vac) .
  \label{eq:woprop}
\end{equation}

\subsection{Relativistic causality from non-relativistic causality}
\label{sec:relcausst}

\begin{figure}
  \centering
  \includegraphics[width=0.6\textwidth]{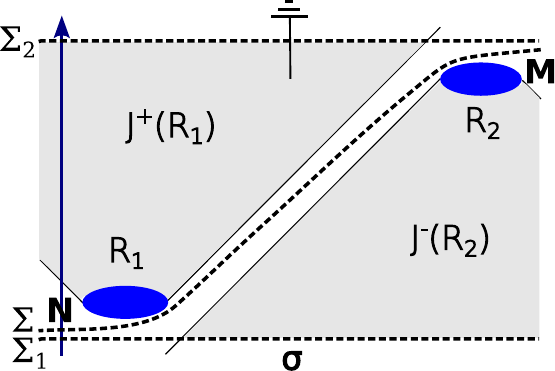}
\caption{Spacetime extended non-selective measurement $N$ in region $R_1$ and selective measurement $M$ in region $R_2$, between initial and final spacelike hypersurfaces $\Sigma_1$ and $\Sigma_2$. A spacelike hypersurface $\Sigma$ separates the two. The initial state at $\Sigma_1$ is $\sigma$ and the system is discarded at $\Sigma_2$.}
\label{fig:relcausality}
\end{figure}

Consider a non-selective probe $\cN_{[\Sigma_1,\Sigma_2]}$ associated to an interval region. We may translate the non-relativistic causality axiom \eqref{eq:nrcaus} to the present setting as,
\begin{equation}
  \discard_{\Sigma_2^\lf}\circ\cN_{[\Sigma_1,\Sigma_2]}=\discard_{\Sigma_1^\lf} .
  \label{eq:nrcst}
\end{equation}
Crucially, if this relation holds, the corresponding relation then holds for the same probe associated to any other interval region. That is, supposing that $\cN_{[\Sigma_1,\Sigma_2]}$ is localizable in a region $R$ we have $\Sigma_1\le R\le\Sigma_2$. Take another interval region $[\Sigma_1,\Sigma_2']$ such that $\Sigma_1'\le R\le \Sigma_2'$. We recall that $\cN_{[\Sigma_1',\Sigma_2']}$ is well-defined, encoding the "same" probe. Using the invariance property \eqref{eq:discext} of the discard probe we then find that the same identity \eqref{eq:nrcst} holds, but with $\Sigma_1,\Sigma_2$ replaced by $\Sigma_1',\Sigma_2'$.

What is more, it is easy to see that relativistic causality then follows from non-relativistic causality. Consider the setting of Figure~\ref{fig:relcausality}, generalizing the instantaneous measurements of Figure~\ref{fig:stlocality} from Section~\ref{sec:relcaus}. We have an initial state $\sigma$ on an initial spacelike hypersurface $\Sigma_1$ and a discard after a final spacelike hypersurface $\Sigma_2$. In between, two measurements take place. A non-selective measurement $N$ is localized in a spacetime region $R_1$ and a selective measurement $M$ is localized in a spacetime region $R_2$. Crucially, $R_2$ is outside of the causal future of $R_1$. This implies that we can choose a spacelike hypersurface $\Sigma$, such that $\Sigma_1\le R_2\le \Sigma\le R_1\le \Sigma_2$. Using the freedom to choose the hypersurfaces delimiting the regions to which the probes encoding the measurements are associated, the probability for a positive outcome is,
\begin{equation}
  \discard_{\Sigma_2^\lf}\circ \cN_{[\Sigma,\Sigma_2]}\circ\cM_{[\Sigma_1,\Sigma]}(\sigma) .
\end{equation}
Here $\cN_{[\Sigma,\Sigma_2]}$ is the probe encoding measurement $N$ and $\cM_{[\Sigma_1,\Sigma]}$ is the probe encoding measurement $M$. However, with the corresponding version of the identity \eqref{eq:nrcst}, this is equal to,
\begin{equation}
  \discard_{\Sigma^\lf}\circ\cM_{[\Sigma_1,\Sigma]}(\sigma)
   = \discard_{\Sigma_2^\lf}\circ\cM_{[\Sigma_1,\Sigma_2]}(\sigma) .
\end{equation}
In particular, the result is independent of measurement $N$, i.e., relativistic causality is satisfied.

\subsection{Measuring spacetime observables}
\label{sec:measurestobs}

\begin{figure}
  \centering
  \includegraphics[width=0.6\textwidth]{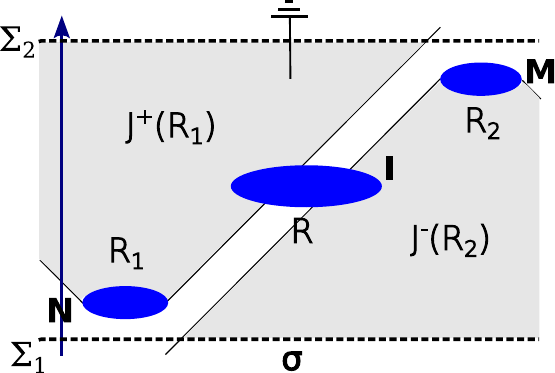}
\caption{Same setup as in Figure~\ref{fig:relcausality}, but with an additional non-selective intermediate measurement $I$ in the spacetime region $R$ to test causal transparency.}
\label{fig:caustrans}
\end{figure}

Even through we have already seen that it does not guarantee causal transparency, we generalize first the \emph{coarse-grained} measurement scheme with discrete outcomes (Section~\ref{sec:nrobs}) to the spacetime setting. We recall that the quantum operation \eqref{eq:selqsubp} for encoding the probability of the outcome being in a subset $X\subseteq S$ of the spectrum $S$ of a phase space observable $A$ can be written as,
\begin{equation}
  \cA_X(\sigma)=\widehat{\chi_X(A)}\,\sigma\,\widehat{\chi_X(A)} .
\end{equation}
Here, $\widehat{\chi_X(A)}$ is the self-adjoint operator associated by quantization to the phase space observable $\chi_X(A)$ obtained as the composition of the characteristic function $\chi_X$ of $X$ with $A$. The corresponding object in the spacetime setting is the correlation function $\rho[\chi_X(A)]$, where $A$ is now taken to be a spacetime observable. This can be obtained in terms of the generalized regularized spectral measure \eqref{eq:spcor} via,
\begin{equation}
  \rho[\chi_X(A)](\ncoh_\phi)=\lim_{\epsilon\to 0} \rho[H_A^\epsilon[\chi_X]](\ncoh_\phi) .
\end{equation}
The corresponding object to encode the measurement in the spacetime setting is the \emph{selective} primitive probe $\cP[\chi_X(A)|\chi_X(A)](\sigma)$. Given a countable partition $\{X_k\}_{k\in I}$ of $\R$ in terms of disjoint non-empty subsets, the corresponding \emph{non-selective probe} should be,
\begin{equation}
  \sum_{k\in I} \cP[\chi_{X_k}(A)|\chi_{X_k}(A)](\sigma) .
\end{equation}
Unfortunately, if $A$ is not a slice observable, it is not clear that this probe satisfies even the \emph{non-relativistic causality axiom} \eqref{eq:nrcaus}. On the other hand we already know that even if $A$ is a slice observable this probe does not in general satisfy \emph{causal transparency}. We will therefore not pursue this measurement scheme further.

We proceed to consider the spacetime generalization of the \emph{fine-grained} measurement scheme with \emph{continuous outcomes}. The spacetime analog of the quantum operation \eqref{eq:rsmop} encoding whether the outcome of measuring the linear spacetime observable $A$ is $q$, is thus the primitive probe,
\begin{equation}
  \cA^\epsilon_M(q)(\sigma)\defeq\sqrt{2\pi}\epsilon\,\cP_M[H_A^\epsilon(q)|H_A^\epsilon(q)](\sigma) .
  \label{eq:probeq}
\end{equation}
This is precisely an example of the \emph{modulus-square construction} \cite{OeZa:lcmeasure}. We also define the corresponding probe-valued functional calculus in analogy to expression \eqref{eq:fcalcqop},
\begin{equation}
  \cA^\epsilon_M[f](\sigma)
  \defeq\int_{-\infty}^{\infty}\xd q\, f(q) \cA^\epsilon_M(q)(\sigma) .
  \label{eq:stopcalc}
\end{equation}
The non-selective probe for this measurement is $\cA^\epsilon_M[\one]$. Crucially, this satisfies the non-relativistic (and thus, by Section~\ref{sec:relcausst} also the relativistic) causality axiom, even if $A$ is not a slice observable, see Appendix~\ref{sec:mathst} for the proof.

\begin{prop}
\label{prop:nsprobeobs}
Let $[\Sigma_1,\Sigma_2]$ be a spacetime region delimited by an initial and a final spacelike hypersurface. Let $A:K_{[\Sigma_1,\Sigma_2]}\to\R$ be a linear observable, $\epsilon>0$ and $\sigma\in\cT_{\Sigma_1}$. Then,
\begin{equation}
  \discard_{\Sigma_2^\lf}\circ\cA_{[\Sigma_1,\Sigma_2]}^\epsilon[\one](\sigma)
  =\discard_{\Sigma_1^\lf}(\sigma) .
\end{equation}
\end{prop}

It remains to verify that the fine-grained measurement scheme of spacetime observables satisfies \emph{causal transparency}. This is illustrated for the present spacetime setting with temporally extended observables in Figure~\ref{fig:caustrans}. For the special case of slice observables this follows from Theorem~\ref{thm:ctrans}. The answer is affirmative also for temporally extended observables as implied by the following theorem, see Appendix~\ref{sec:mathst} for the proof.

\begin{thm}
  \label{thm:ctst}
  Let $R_1,R,R_2$ be disjoint regions such that $R_2$ does not intersect the causal future of $R_1$. Let $\cN$ be a non-selective probe localizable at $R_1$, $\cM$ a selective probe localizable at $R_2$ and $A:K_{R}\to\R$ a linear observable. Let $\Sigma_1,\Sigma_2$ be spacelike hypersurfaces $\Sigma_1\le\Sigma_2$, such that $R_1,R,R_2$ are all in the future of $\Sigma_1$ and in the past of $\Sigma_2$. Let $\epsilon>0$ and $\cI=\cA^{\epsilon}[\one]$ and $\sigma\in\cT_{\Sigma_1}$. Then,
  \begin{equation}
     \discard_{\Sigma_2^\lf}\circ\left(\cM\comp \cI\comp \cN\right)(\sigma)
     = \discard_{\Sigma_2^\lf}\circ\left(\cM\comp \cI\right)(\sigma).
  \end{equation}
  The probes $\cM,\cN,\cI$ are associated to regions that are intervals and together comprise $[\Sigma_1,\Sigma_2]$.
\end{thm}

As for everything else, the fine-grained measurement scheme of the instantaneous operator setting (Section~\ref{sec:instqft}) is recovered by using slice observables. As there, we may ask for a generalization to non-linear observables. The answer is also similar. In particular, the proof of the causal transparency property relies on the linearity of the observable. But the second part of the answer given in that setting applies also in the spacetime generalization: Namely, we can model the measurement of non-linear observables that arise as functions of linear observables within the (regularized) \emph{operational functional calculus} given by expression \eqref{eq:stopcalc}. It is instructive however, to get some understanding of the difference to a naive insertion of a non-linear observable $A$ into expression \eqref{eq:aeobs}. Suppose that $A$ is a linear observable, but we want to measure the expectation value of $g(A)$, where $g:\R\to\R$ is a continuous function that increases less than exponentially. With the insertion method we obtain the (regularized) probe
\begin{equation}
  \sqrt{2\pi}\epsilon \int_{-\infty}^{\infty}\xd q\, q\, \cP[H_{g(A)}^\epsilon|H_{g(A)}^\epsilon] .
  \label{eq:stnlinins}
\end{equation}
In contrast, the operational calculus \eqref{eq:stopcalc} yields,
\begin{equation}
  \sqrt{2\pi}\epsilon \int_{-\infty}^{\infty}\xd q\, g(q)\, \cP[H_{A}^\epsilon|H_{A}^\epsilon] .
  \label{eq:stnlincalc}
\end{equation}
Both of these are calculated in terms of the Schwinger-Keldysh double path integral with observables inserted. By thinking of the latter as observables $K\times K\to\R$ on the double field configuration space we can make a comparison without attempting to fully evaluate the path integrals. Furthermore, we assume $\epsilon$ to be small to evaluate the integrals over $q$. In the insertion case \eqref{eq:stnlinins} we obtain,
\begin{equation}
  (\phi,\phi')\mapsto \exp\left(-\frac{1}{2\epsilon^2}\left(g(A(\phi))-g(A(\phi'))\right)^2\right)
  \frac12 \left(g(A(\phi))+g(A(\phi'))\right) .
\end{equation}
Here we label the field configurations in the forward and backward path integrals $\phi$ and $\phi'$ respectively.
In contrast, from the operational calculus \eqref{eq:stnlincalc} we obtain,
\begin{equation}
  (\phi,\phi')\mapsto \exp\left(-\frac{1}{2\epsilon^2}\left(A(\phi)-A(\phi')\right)^2\right)
   g\left(\frac12 \left(A(\phi)+A(\phi')\right)\right) .
\end{equation}
In both cases, the first factor acts like an approximate delta-function while the second factor under that constraint yields approximately $g(A(\phi))$, as required. There is an interesting difference, however. In the first case, the delta function would enforce $g(A(\phi))=g(A(\phi'))$, while in the second $A(\phi)=A(\phi')$. While the second equality implies the first, this is not the case the other way round. More precisely, it is the case the other way round if and only if $g$ is invertible. This is in line with earlier considerations about the fine-grained measurement scheme \cite{Oe:spectral}, see also Section~\ref{sec:instqft}. For causal transparency (satisfied by the non-selective version of the probe \eqref{eq:stnlincalc} but not of the probe \eqref{eq:stnlinins}) it seems essential that we extract from the quantum system the information of the complete value of the observable (here $A$) and not only a reduced version of that information (here $g(A)$ if $g$ is not invertible).

%% file: multiobs.tex

\section{Composite Observables}
\label{sec:multiobs}

As in the operator setting a straightforward generalization to a multi-observable calculus can be performed in the spacetime setting. Consider linear observables $A_1,\ldots,A_n$ localizable in disjoint regions, and a regulator $\epsilon>0$. We may then define the composite primitive probe,
\begin{equation}
  \cA^\epsilon[\va{q}]\defeq\cA_1^\epsilon(q_1)\comp\cdots\comp\cA_n^\epsilon(q_n)= (\sqrt{2\pi} \epsilon)^n \cP[H_{A_1}^\epsilon(q_1)\cdots H_{A_n}^\epsilon(q_n)|H_{A_1}^\epsilon(q_1)\cdots H_{A_n}^\epsilon(q_n)],
   \label{eq:stcomprprobe}
\end{equation}
where we abbreviate $\va{q}\defeq (q_1,\ldots,q_n)$. Given a function $f:\R^n\to\R$ we further define,
\begin{align}
  \cA^\epsilon[f] & \defeq
   \int_{\R^n}\xd \va{q}\, f(\va{q})\,\cA^\epsilon[\va{q}] .
   \label{eq:probef}
\end{align}
The corresponding non-selective probe is $\cA^\epsilon[\one]$, where $\one$ is again the constant function with unit value. Relativistic causality and causal transparency follow by composition from those properties for the individual probes $\cA_k^\epsilon[\one]$.\footnote{Note a subtlety here: By using relativistic causality \emph{and} causal transparency of the individual component probes, we can derive relativistic causality of the composite probe, \emph{without} requiring that the component probes are localizable in regions that are globally causally orderable.} The right-hand side of expression \eqref{eq:stcomprprobe} makes the freedom of composition in spacetime that our formalism affords rather explicit, through the commutative multiplication of functions (observables) on field configuration space. What is more, this expression is in principle also well-defined if the component observables are localized in spacetime regions that overlap. However, it is unclear whether this would be physically sensible. We will not pursue this possibility further.

It turns out that quite explicit formulas can be derived for probabilities and expectation values of these composite probes. This is the focus of the present section. Proofs of statements in this section can be found in Appendix~\ref{sec:pmultiobs}. We introduce the following notation, $W_{i j}\defeq \wuu(A_i,A_j)$ and use $W$ to denote the matrix with entries $W_{i j}$. We also write $W=R+\im I$, where $R$ is the real and $I$ is the imaginary part of the matrix $W$, understood entry-wise. Since the matrix entries are the Feynman propagators this means that $R$ consists of Hadamard propagators while $I$ consists of causal propagators. More precisely, $2I=W_{\mathrm{ret}}+W_{\mathrm{adv}}$, where $W_{\mathrm{ret}}$ is the matrix consisting of retarded propagators, while $W_{\mathrm{adv}}$ is the matrix consisting of advanced propagators, compare Section~\ref{sec:amocor}. Recall that the Hadamard propagator is positive. It is easy to see that this implies that the matrix $R$ viewed as an operator on a complex Hilbert space of dimension $n$ is positive, i.e., $R\ge 0$.

As a first step we note that Proposition~\ref{prop:regobs} generalizes to multiple observables. We write $\va{\phi}\defeq (A_1(\phi^\Lint),\ldots,A_n(\phi^\Lint))$. On coherent states we obtain the following value.
\begin{prop}
\label{prop:multicor}
For $\epsilon>0$ the matrix $\epsilon^2\id+2W$ is invertible, and the following is well-defined.
\begin{multline}
   \rho[H_{A_1}^{\epsilon}(q_1)\cdots H_{A_n}^{\epsilon}(q_n)](\ncoh_{\phi}) \\
   =\frac{1}{\sqrt{\pi^n \det(\epsilon^2\id + 2W)}}\rho(\ncoh_{\phi})
      \exp\left(- \left(\va{\phi}-\va{q}\right)^\mathrm{T}
       \left(\epsilon^2\id+2W\right)^{-1}
      \left(\va{\phi}-\va{q}\right) \right) .
\end{multline}
Moreover, if $R$ is non-degenerate, then $W$ is invertible and the expression is well-defined also for $\epsilon=0$.
\end{prop}
We may infer the value of the composite probe at outcome $\va{q}$ from this correlation function. To this end it is convenient to evaluate on a generalized pure state $\Xi_{\beta|\gamma}\defeq |\ncoh_\beta\rangle\langle\ncoh_{\gamma}|$ since the definition \eqref{eq:probecorpure} then implies a factorization,
\begin{equation}
  \cP[F|G](\Xi_{\beta|\gamma})=
  \rho[F](\ncoh_{\beta}) \overline{\rho[G](\ncoh_{\gamma})} .
\end{equation}
The use of these generalized pure states $\Xi_{\beta|\gamma}$, which are not really states, but objects in the vector space $\cB$ spanned by states, is also convenient since they can be used to generate all states via (two copies of) the completeness relation \eqref{eq:cohcompl}. This is not possible with ordinary coherent states $\Xi_{\phi}\defeq\Xi_{\phi|\phi}=|\ncoh_\phi\rangle\langle\ncoh_{\phi}|$.
With the identity \eqref{eq:stcomprprobe} we then obtain from Proposition~\ref{prop:multicor} the following result.
\begin{prop}
   \label{prop:multiprobeq}
\begin{multline}
\cA^\epsilon[\va{q}](\Xi_{\beta|\gamma})
   =\np(\Xi_{\beta|\gamma})\sqrt{\frac{2}{\pi}}^n
   \frac{\epsilon^n}{|\det(\epsilon^2\id + 2W)|} \\
      \exp\left(- \left(\va{\beta}-\va{q}\right)^\mathrm{T}
       \left(\epsilon^2\id+2W\right)^{-1}
      \left(\va{\beta}-\va{q}\right) 
      -\left(\overline{\va{\gamma}}-\va{q}\right)^\mathrm{T}
       \left(\epsilon^2\id+2\overline{W}\right)^{-1}
      \left(\overline{\va{\gamma}}-\va{q}\right) 
      \right) .
\end{multline}
\end{prop}
This is well-defined for $\epsilon>0$. In case $R$ is non-degenerate this vanishes for $\epsilon=0$. This reflects the singular nature of the limit where the regulator $\epsilon$ is removed.
We proceed to consider the integrated form of the composite regularized observable probe. To this end we take the observable function to be a generating function of the form $w_{\va{s}}(q)\defeq\exp(\im\, \va{s}\cdot\va{q})$, where $\va{s}=(s_1,\ldots,s_n)$ are real numbers. Here, $\va{s}\cdot\va{q}=\sum_{i=1}^n s_i q_i$ is the dot product of vectors. 
\begin{prop}
   \label{prop:mprobeweyl}
\begin{multline}
  \cA^\epsilon[w_{\va{s}}](\Xi_{\beta|\gamma})
   =\np(\Xi_{\beta|\gamma})
   \frac{\epsilon^n}{\sqrt{\det(\epsilon^2\id + 2R)}}
   \exp\left(-\frac12 \left(\va{\beta}-\overline{\va{\gamma}}\right)^{\mathrm{T}}
   \left(\epsilon^2\id+2R\right)^{-1}
   \left(\va{\beta}-\overline{\va{\gamma}}\right)
   \right)\\
   \exp\left(\im\, \va{s}\cdot\frac{\va{\beta}+\overline{\va{\gamma}}}{4}+\va{s}^\mathrm{T}  I \left(\epsilon^2\id+2R\right)^{-1}\frac{\va{\beta}-\overline{\va{\gamma}}}{2}\right)
   \exp\left(-\frac18 \va{s}^{\mathrm{T}}\left(\epsilon^2\id+2R
    + 4 I\left(\epsilon^2\id+2R\right)^{-1} I\right) \va{s}\right) .
\end{multline}
\end{prop}

We proceed to consider the case where we discard the system in the future. In the following we change notation to write $\va{\phi}\defeq (A_1(\phi),\ldots,A_n(\phi))$.
\begin{prop}
   \label{prop:mprobeqdisc}
   Define
\begin{equation}
X_{\epsilon}\defeq \frac{\epsilon^2}{2}\id + 2R +\frac{2}{\epsilon^2} W_{\mathrm{ret}} W_{\mathrm{adv}} .
\label{eq:defXe}
\end{equation}
This matrix is invertible for $\epsilon>0$ since in addition to $R$, the product $W_{\mathrm{ret}} W_{\mathrm{adv}}$ is also positive. (Indeed $W_{\mathrm{ret}}^\dagger=W_{\mathrm{adv}}$.) For $\epsilon=0$ this is well-defined only if $W_{\mathrm{ret}} W_{\mathrm{adv}}=0$, which implies $W_{\mathrm{ret}}=0=W_{\mathrm{adv}}$. It is then moreover invertible if in addition $R$ is non-degenerate.
\begin{equation}
  \discard\circ\cA^\epsilon[\va{q}](\Xi_{\beta|\gamma})
   =\discard(\Xi_{\beta|\gamma})\frac{1}{\sqrt{\pi^n \det X_{\epsilon}}}
   \exp\left(-\left(\va{\beta}^- +\va{\gamma}^+ - \va{q}\right)^{\mathrm{T}}
   X_{\epsilon}^{-1}
   \left(\va{\beta}^- +\va{\gamma}^+ - \va{q}\right)
   \right) .
\end{equation}
\end{prop}
Note that the condition $W_{\mathrm{ret}}=0$ (or equivalently $W_{\mathrm{adv}}=0$) is satisfied if all observables are spacelike separated with respect to each other and also with respect to themselves. The latter implies in particular that they are slice observables and that there exists a spacelike hypersurface on which all observables are localizable.

The following result is easily deduced by performing a Gaussian integral. Moreover, it takes the form of a factorization identity.
\begin{prop}
   \label{prop:mprobeweyldisc}
\begin{equation}
  \discard\circ\cA^\epsilon[w_{\va{s}}](\Xi_{\beta|\gamma})
   =\discard(\Xi_{\beta|\gamma})
   \exp\left(\im\, \va{s}\cdot \left(\va{\beta}^- +\va{\gamma}^+\right)\right)
   \exp\left(-\frac14 \va{s}^{\mathrm{T}} X_{\epsilon} \va{s}\right) .
\end{equation}
\end{prop}
As the special case $\va{s}=\va{0}$ we recover the non-relativistic causality condition for the non-selective probe, $\discard\circ\cA^\epsilon[\one](\sigma)=\discard(\sigma)$, generalizing Proposition~\ref{prop:nsprobeobs} to the multi-observable setting.

For a single observable and a proper coherent state, we reproduce and generalize the following semiclassical result of the operator setting \cite{Oe:spectral}.
\begin{prop}
   Let $\lambda,c\in\R$. Set $f(q)=\lambda q +c$. Then,
\begin{equation}
  \discard\circ\cA^\epsilon[f](\Xi_{\beta})=f(A(\beta)) .
\end{equation}
\end{prop}

%% file: causalcor.tex

\section{Causal correlations in relativistic quantum measurement}
\label{sec:causalcor}

It is striking how the formulas we have obtained reveal the influence of the causal structure of spacetime on measurement correlations in relativistic quantum theory. The present section presents a very preliminary and limited consideration of this question. A proper analysis is out of scope for the present paper.

At first, we consider a scenario where we can demonstrate causal transparency directly by using the expression for multi-observable probes of Proposition~\ref{prop:mprobeweyldisc}. Thus, we consider three linear observables $A_1,A_2,A_3$ localizable in regions $R_1,R,R_2$ respectively, compare Figure~\ref{fig:caustrans}. The important constraint is that $R_1$ does not intersect the causal past of $R_2$, or equivalently $R_2$ does not intersect the causal future of $R_1$. Crucially this implies $\wr(A_3,A_1)=0$ and equivalently $\wa(A_1,A_3)=0$. Our measurement of $A_1$ and of $A_2$ is non-selective while our measurement of $A_3$ is selective. 
Thus, we are precisely in a setting as considered in Section~\ref{sec:measurestobs} and in Theorem~\ref{thm:ctst}, except our measurements in $R_1$ and $R$ are of a more specific type. If we assign values $q_1,q_2,q_3$ to $A_1, A_2, A_3$, non-selectiveness for $A_1$ and $A_2$ means that the function determining the composite probe may neither depend on $q_1$ nor on $q_2$. We take the measurement to consist in comparing the value of $A_3$ with a fixed value $a$, i.e.\ we set $f_a(q_1,q_2,q_3)=\delta(q_3-a)$. With a pure coherent state $\Xi_{\phi}=\Xi_{\phi|\phi}$ as the initial state we obtain the desired probe from Proposition~\ref{prop:mprobeweyldisc} by setting $s_1=s_2=0$ and performing a Fourier transform in $s_3$. This yields,
\begin{align}
  & \discard\circ\cA^\epsilon[f_a](\Xi_{\phi})
   =\frac{1}{2\pi}\int_{-\infty}^{\infty}\xd s_3\,e^{-\im s_3 a}\discard\circ\cA^\epsilon[w_{(0,0,s_3)}](\Xi_{\phi}) \nonumber \\
  & \qquad 
   =\frac{1}{\sqrt{\pi (X_{\epsilon})_{3 3}}}
   \exp\left( -\frac{1}{(X_{\epsilon})_{3 3}}
   \left(A_3(\phi)-a\right)^2\right), 
   \label{eq:3linct} \\
  & \text{with}\quad (X_{\epsilon})_{3 3}
  =\frac{\epsilon^2}{2}+ \wh(A_3,A_3) +\frac{2}{\epsilon^2}\wr(A_3,A_3)\wa(A_3,A_3)
   +\frac{2}{\epsilon^2}\wr(A_3,A_2)\wa(A_2,A_3) .
   \label{eq:3linXe}
\end{align}
Recall that $X_{\epsilon}$ is a matrix and $(X_{\epsilon})_{3 3}$ denotes its $(3,3)$ entry. There is an additional term in the definition (\ref{eq:defXe}) of $(X_{\epsilon})_{3 3}$ of the form $\wr(A_3,A_1)\wa(A_1,A_3)$. However, as already mentioned, this vanishes due to the causal relation between $R_1$ and $R_2$. It remains to realize that the obtained expressions \eqref{eq:3linct} and \eqref{eq:3linXe} are identical to the ones we would have obtained if only $A_2$ and $A_3$ had been measured, but not $A_1$. That is, all probabilities or expectation values of measuring $A_3$ are the same, whether we jointly measure with both $A_1$ and $A_2$ or just with $A_2$, thus
demonstrating causal transparency of the measurement of $A_2$.

We proceed to consider correlation functions for the measurement of $n$ linear observables. Let $A_1,\ldots,A_n$ be the linear observables with associated values $q_1,\ldots,q_n$. We assume an initial coherent state $\Xi_{\phi}$. The correlation function is then determined by the function $f(q_1,\ldots,q_n)=q_1\cdots q_n$. The associated probe can be deduced from Proposition~\ref{prop:mprobeweyldisc} as follows:
\begin{align}
  & \discard\circ\cA^\epsilon[f](\Xi_{\phi})
  =(-\im)^n \frac{\partial}{\partial s_1}\cdots\frac{\partial}{\partial s_n}\discard\circ\cA^\epsilon[w_{\va{s}}](\Xi_{\phi}) \Big|_{s_1=\cdots =s_n=0} 
  \nonumber \\
  & =\sum_{m=0}^{\lfloor n/2 \rfloor} \sum_{\sigma\in S^n} \frac{1}{2^m\, m!\, (n-2m)!} A_{\sigma(2m+1)}(\phi)\cdots A_{\sigma(n)}(\phi)\, \prod_{j=1}^m\, \left(\frac12 X_{\epsilon}\right)_{\sigma(2j-1), \sigma(2j)} .
  \label{eq:mobscor}
\end{align}
Here $\lfloor a\rfloor$ denotes the ``floor'' of $a$, i.e.\ the largest integer smaller or equal to $a$. This looks very much like the usual Wick theorem: We sum over all possible ways to pair up some of the observables though propagators, while evaluating the remaining ones on the solution corresponding to the initial coherent state. By way of comparison, the usual time-ordered correlation function for the observables $A_1,\ldots,A_n$ with initial state $\Xi_{\phi}=|\ncoh_{\phi}\rangle\langle\ncoh_{\phi}|$ and discard can be written as,
\begin{equation}
  \langle\ncoh_\phi,\tord A_1\cdots A_n \ncoh_\phi\rangle
  =\sum_{m=0}^{\lfloor n/2 \rfloor} \sum_{\sigma\in S^n} \frac{1}{2^m\, m!\, (n-2m)!} A_{\sigma(2m+1)}(\phi)\cdots A_{\sigma(n)}(\phi)\, \prod_{j=1}^m\, W_{\sigma(2j-1), \sigma(2j)} . \label{eq:tocor}
\end{equation}
Here, the matrix elements of $W$ are the observable Feynman propagators, i.e., $W_{i j}=\wuu(A_i,A_j)$. (In the case of scalar field theory and for point observables $A_k(\phi)=\phi(t_k,x_k)$ these are of course the usual Feynman propagators $\wtuu(t_i,x_i,t_j,x_j)=\langle\vac,\tord \phi(t_i,x_i)\phi(t_j,x_j)\vac\rangle$.)

When comparing the correlation functions it is important to recall why they are different. The time-ordered correlation function \eqref{eq:tocor}, while a very important object in QFT, for example as an ingredient for constructing the S-matrix, does not itself correspond to the expectation value of any measurement. In contrast, the measurement correlation function \eqref{eq:mobscor} encodes precisely the expectation value of the joint measurement of the observables as localized in spacetime. Clearly, the Feynman propagator, being complex, cannot appear in a measurement of correlations of real observables. Rather, for any pairing between distinct observables appearing in the correlation function the Feynman propagator 
$\wuu(A_i,A_j)$ is replaced by the matrix element ($i\neq j$)
\begin{equation}
  \left(\frac12 X_{\epsilon}\right)_{i j}= \wh(A_i,A_j) + \frac{1}{\epsilon^2} \sum_{k=1}^n \wr(A_i,A_k)\wa(A_k,A_j) .
\end{equation} 
This is real as it should be. What is more, the first term is precisely the Hadamard propagator, i.e.\ the real (acausal) part of the Feynman propagator as one might have expected. However, there is a second term which is highly significant. The retarded propagator $\wr(A_i,A_k)$ is non-vanishing only if the spacetime support of $A_i$ intersects the causal future of the support of $A_k$. Similarly, the advanced propagator $\wa(A_k,A_j)$ is non-vanishing only if the support of $A_j$ intersects the causal future of the support of $A_k$. That is, the second term arises only if $A_i$ and $A_j$ both have support in the causal future of the support of $A_k$. This suggests that the measurement of $A_k$ induces an additional correlation between $A_i$ and $A_j$. Moreover, this correlation is not related to the measured value of $A_k$ (assuming here that $k\neq i$ and $k\neq j$), as the function $f$ determining the outcome can be modified to not depend on $q_k$ without affecting the appearance of this correlation term. That is, the correlation between $A_i$ and $A_j$ is caused purely be the "disturbance" originating from the measurement of $A_k$. The signal emanating from the disturbance affects equally $A_i$ and $A_j$ which are both in its causal future, hence the correlation. Of course the disturbance from the measurement of $A_k$ affects the measurement of all observables with support in its causal future. And indeed, it causes correlations between any pair with supports in this causal future in this same way. However, there are no higher-order correlations. Rather, all these correlations factorize into pairs as might be expected in the light of Sorkin's analysis on the absence of higher-order correlations in quantum theory \cite{Sor:qmeasure}.

The disturbance caused by a measurement does not only affect other measurements and their correlation, but also the measurement itself. That is, the earlier part of a measurement "back-reacts" on the later part, as long as the measurement is temporally extended, i.e., not localizable on a spacelike hypersurface. While this effect is clearly present in the correlation functions just discussed, it is most cleanly seen by considering the measurement of a single linear observable $A$ alone. The relative probability density given by Proposition~\ref{prop:mprobeqdisc} depends on the term $X_{\epsilon}$, which in this case is,
\begin{equation}
  X_{\epsilon}=\frac{\epsilon^2}{2}+ 2\wh(A,A) + \frac{2}{\epsilon}(\wr(A,A))^2 .
\end{equation}
It is the causal term $\wr(A,A)=\wa(A,A)=\Im(\wuu(A,A))$ that arises from the temporal extension of $A$. It disappears if we replace $A$ by a slice observable $A_{\Sigma}$ that is identical on the classical phase space, compare Section~\ref{sec:sobs}.

Another interesting aspect of the disturbances and correlations caused by the measurements is that they grow when the regulator $\epsilon$ is sent to zero, dominating over the Hadamard propagators that may be interpreted to encode correlations arising from the vacuum. That is, the "sharper" the measurements are taken to be performed, the stronger are the correlations caused by the measurements. This divergence (when $\epsilon\to 0$) appears in addition to, but is qualitatively different from, the usual divergence that arises in the Hadamard propagator for point observables in the coincidence limit.

%% file: outlook.tex

\section{Discussion and Outlook}
\label{sec:outlook}

This paper should first of all be read as a paper on \emph{quantization}. We start this section by elaborating on this perspective on the present work.
Quantization refers here to a prescription that assigns to a classical observable an object in the quantum theory that encodes a \emph{measurement} corresponding to that classical observable.
The traditional narrative on quantization consists of two components: The first is the idea that self-adjoint operators are the objects that determine measurements of observable quantities in quantum mechanics. The second is that there are procedures converting classical observables to self-adjoint operators in such a way that the induced measurement corresponds in some reasonable way to the measurement of the classical observable. Such procedures are known as \emph{quantization prescriptions}. The first of these components has been refined over time, as it was understood that more general measurements on a quantum system can be performed, leading to the modern notion of a \emph{quantum operation}. It is thus reasonable to adapt the second component as well. That is, we shall talk about quantization prescriptions as procedures that convert a classical observable to a (set of) quantum operations that encode the measurement of the classical observable in a suitable sense. To distinguish this from the traditional notion of quantization we refer to this as \emph{operational quantization}. In this context, the path from traditional quantization to operational quantization appears straightforward. As we have recalled in Section~\ref{sec:nrquant}, the spectral decomposition of the self-adjoint operator obtained by (traditional) Weyl quantization serves to construct the desired quantum operations. While in this case operational quantization "passes through" traditional quantization, this need not be the case, as we shall see later on. Going a step further, one might want to establish criteria for a "good" operational quantization that do not rely on such an intermediate step, in analogy to the catalogs of criteria that characterize traditional quantization prescriptions. While we will discuss certain such criteria, an attempt at completeness is out of scope for the present work.

When moving from non-relativistic quantum mechanics to QFT there arise additional requirements on measurements from locality and causality conditions (Section~\ref{sec:loccaus}). Crucially, these requirements can be translated into criteria for an operational quantization prescription, i.e., one targeting quantum operations. However, only part of these criteria can be made sense of at the level of the intermediate Weyl quantization. More concretely, it does make sense to talk about the locality properties of field operators (arising from Weyl quantization), but it makes sense to talk about causal transparency only at the level of the quantum operations.
What is more, the operational quantization prescription for field operators laid out in Section~\ref{sec:instqft} modifies the role that the intermediate Weyl quantization plays. While the conceptual starting point is the spectral decomposition of the field operator, a regularization turns out to be necessary. As a consequence, the regularized quantum operations involve operators that are (traditional) quantizations of more complicated functions of the observable \eqref{eq:aeobs}. It also turns out that the regulator cannot in general be removed at the level of quantum operations as we will detail later.

When moving from a time-evolution setting to a genuine spacetime setting the accompanying notions of quantization change even more radically (Section~\ref{sec:spacetime}). To be able to accommodate observables that are extended in time as well as in space, and moreover freely compose corresponding measurements, we need to replace the standard formulation of quantum theory with the more powerful local positive formalism (PF) \cite{Oe:dmf,Oe:posfound}. With this, instantaneous quantum operations are generalized to spacetime localized \emph{probes}. These are now the targets of operational quantization prescriptions. At the same time, the classical observables to be quantized are no longer observables on phase space, but observables on field configuration space. This is because the role of the intermediate traditional quantization prescription is now played by path integral quantization. The probes are then constructed through a double path integral in a Schwinger-Keldysh formalism \cite{OeZa:lcmeasure}. The quantization prescription we propose for the spacetime setting consists of the obvious spacetime generalization of the prescription arising from the regularized spectral decomposition in the time-evolution setting (Section~\ref{sec:measurestobs}). We use the same induced observable \eqref{eq:aeobs} for a given linear observable of interest. The difference is that the observable lives now on field configuration space rather than phase space and that instead of Weyl quantization we perform path integral quantization (Section~\ref{sec:gspec}). Nevertheless, when we specialize to slice observables (Section~\ref{sec:sobs}), i.e., spacetime observables that arise from phase space observables, we recover exactly the results of the time-evolution setting. Crucially, this recovery applies not only to single measurements, but to arbitrary composites. On the other hand, and this is a main result of the present paper, we still satisfy causal transparency even in the general spacetime setting (Theorem~\ref{thm:ctst}). With respect to the operational quantization prescription we may note that the induced observable \eqref{eq:aeobs} can be seen as an approximation of a delta function (compare Lemma~\ref{lem:did}), enforcing the value of the observable $A(\phi)$ on the field $\phi$ to coincide with the prescribed value $q$. Then, the operational quantization prescription for the probe that asks for the value $q$ consists of the Schwinger-Keldysh double path integral with this approximate delta-function inserted in both branches. Even though the Gaussian approximation we use turns out to be very convenient, this suggests that other choices of approximations of delta-functions might be used. However, for most such choices causal transparency will likely not hold.

Measurements necessarily disturb a quantum system and a temporally extended measurement will therefore necessarily back-react on itself. A main result of the present work is the exhibition and quantification of this effect (in Sections~\ref{sec:multiobs} and \ref{sec:causalcor}). What is more, we have found that the disturbance caused by a measurement also induces correlations between later measurements in its causal future. Measurement back-reaction 
has been of interest in the literature in the context of "continual observation". Remarkably, it is precisely in this context that Gaussian quantum operations analogous to \eqref{eq:gop} appear to have first been proposed \cite{BaLaPr:marcocontobsqm}.
Continual observations are modeled in terms of a temporal sequence of repeated measurements with a limit taken where the number of measurements is increased while the time intervals between measurements are shortened. For measurements of linear observables this corresponds in our spacetime context to the approximation of a temporally extended observable by a sum of slice observables localized on consecutive spacelike hypersurfaces. The notion of convolution of consecutive measurements (in terms of instruments) \cite{Hol:diviquantumprob} corresponds then at the level of observables to the convolution identity of Lemma~\ref{lem:addid}. Moreover, infinite divisibility of instruments (necessary for the existence of a limit) is then related to our notion of additive decomposability of observables in spacetime (Section~\ref{sec:loccomp}).

For measurement back-reaction the difference between spacetime observables and phase space observables and their roles in quantization is crucial. It is common in the literature to see parametrizations of a field operator $\hat{\phi}(f)$ (in scalar field theory) in terms of a spacetime smearing function $f$, compare expression \eqref{eq:stop}. This suggests an observable extended over the spacetime region determined by the support of $f$. Instead, we are dealing with the quantization of an observable on the instantaneous phase space corresponding to the time when the measurement takes place. In particular, the expression $\tr(\hat{\phi}(f)\sigma)$ yields the expectation value of the  measurement of this instantaneous observable (in an initial state $\sigma$).
This is also a fundamental limitation of \emph{algebraic quantum field theory} \cite{HaKa:aqft}. Even though algebras of quantum observables are associated to spacetime regions, these observables are instantaneous as are their expectation values when paired with states. The \emph{time slice axiom} is a direct expression of this circumstance \cite{Haa:lqp}.
To correctly describe the measurement of time-extended observables, their modeling as spacetime observables rather than phase space observables is essential. As we have seen, the whole quantization prescription needs to be adapted as well (Section~\ref{sec:spacetime}). In particular, the simple formula \eqref{eq:evsingled}, $\tr(\hat{A}\sigma)$ for the expectation value of an observable $A$ (in an initial state $\sigma$) is no longer valid for time-extended observables, even if we replace Weyl quantization with path integral quantization.

An obvious question concerning the regularization we have introduced in the construction of quantum operations and probes to measure observables is whether or when it can be removed. We have discussed already in Section~\ref{sec:causalcor} the finding that a measurement causes correlations between other measurements in its causal future. Moreover, this correlation diverges if the regulator $\epsilon$ is sent to $0$. Clearly, in this case we cannot remove the regulator. On the other hand one might attribute the singular behavior to the unboundedness of the weight function $f$ in the probe $\cA^{\epsilon}[f]$. Consider therefore the simplest probe associated with a set of observables, namely the \emph{non-selective one}. By our construction this probe $\cA^{\epsilon}[\one]$ satisfies the non-relativistic causality axiom \eqref{eq:nrcaus} in particular, so $\discard\circ\cA^{\epsilon}[\one](\sigma)=1$ for a normalized state $\sigma$. This has a trivial limit $\epsilon\to 0$. Consider on the other hand, a setting with post-selection. Let $M$ be a spacetime region (for example a time-interval) and $\Xi_{\phi}\in\cB_{\partial M}$ a normalized coherent state on the boundary $\partial M$ of $M$. Then, $\cA^{\epsilon}[\one](\Xi_{\phi})$ has a perfectly well-defined limit $\epsilon\to 0$ if the matrix $R$ is non-degenerate. (Set $\va{s}=\va{0}$ and $\beta=\gamma=\phi$ in Proposition~\ref{prop:mprobeweyl}.) However, the limit is $\lim_{\epsilon\to 0}\cA^{\epsilon}[\one](\Xi_{\phi})=0$. So we have two existing, but different limits for $\cA^{\epsilon}[\one]$, depending on the boundary conditions. This implies, clearly, that $\lim_{\epsilon\to 0}\cA^{\epsilon}[\one]$ does not exist \emph{as a quantum operation}. What is more, we might be worried about the vanishing of $\lim_{\epsilon\to 0}\cA^{\epsilon}[\one](\Xi_{\phi})$ in the post-selection setting. Indeed, $\lim_{\epsilon\to 0}\cA^{\epsilon}[\va{q}](\Xi_{\phi})$ for a given outcome $\va{q}$ vanishes as well. However, we are outside the realm of the standard formulation of quantum theory, where we can use the non-relativistic causality condition \eqref{eq:nrcaus} as a normalization condition. Rather, the probability density for measuring $\va{q}$ is the quotient \cite{Oe:posfound},
\begin{equation}
 \frac{\cA^{\epsilon}[\va{q}](\Xi_{\phi})}{\cA^{\epsilon}[\one](\Xi_{\phi})} .
\end{equation}
As it turns out, this quotient has a well-defined and non-vanishing limit $\epsilon\to 0$. Indeed, suitably scaled, numerator and denominator have individual well-defined and non-vanishing limits. This is $\lim_{\epsilon\to 0}\epsilon^{-n}\cA^{\epsilon}[\va{q}](\Xi_{\phi})$ and $\lim_{\epsilon\to 0}\epsilon^{-n}\cA^{\epsilon}[\one](\Xi_{\phi})$. (We continue to assume that the matrix $R$ is non-degenerate.) We learn that the proper scaling of the limit, if it exists, depends in general on the boundary conditions of the measurement. In any case, the correct setting for a limit is that of physical probabilities and expectation values, which can always be expressed in terms of a quotient \cite{Oe:posfound}.

We have presented in Section~\ref{sec:multiobs} a scheme to construct probes for measuring functions that depend on multiple linear observables in a rather general way. To illustrate this with an example of particular physical interest, we consider the energy momentum tensor of real scalar field theory \cite{BiDa:qftcurved},\footnote{For simplicity, we use a notation in this paragraph where $x$ denotes a point is spacetime, rather than in space.}
\begin{equation}\label{eq:emtensor}
T_{\mu\nu} (x)[\phi]=\partial_{\mu}\phi(x) \partial_{\nu}\phi(x)
-\frac12 g_{\mu \nu}(x) g^{\rho \tau}(x) \partial_{\rho}\phi(x) \partial_{\tau}\phi(x) + \frac12 m^2 g_{\mu \nu}(x) .
\end{equation}
To model this, we use five linear observables, $\phi(x),\partial_0\phi(x),\ldots,\partial_3\phi(x)$ with respective values $q,q_0,\ldots,q_3$. The function to measure the expectation value of $T_{\mu \nu}(x)$ is,
\begin{equation}
  f_{\mu \nu}(q,q_0,\ldots,q_3)\defeq q_{\mu} q_{\nu}
-\frac12 g_{\mu \nu}(x) g^{\rho \tau}(x) q_{\rho} q_{\tau} + \frac12 m^2 g_{\mu \nu}(x) q^2 .
\end{equation}
Then, the probe that encodes this expectation value is $\cA^\epsilon[f_{\mu \nu}]$. As always, the corresponding non-selective probe is $\cA^{\epsilon}[\one]$. Of course, instead of an expectation value we could also ask for the probability for the value of the energy-momentum tensor to lie in a certain range, say the interval $[a,b]$. The probe encoding this is then $\cA^\epsilon[\chi_{[A,B]}\circ f_{\mu \nu}]$, i.e., we simply compose the expectation function $f_{\mu \nu}$ with the characteristic function of the interval $[a,b]$. We shall not pursue this example further here as it requires a renormalization, due to the quadratic nature of the energy-momentum observable at a spacetime point. However, we remark that this could be performed in different ways. One would be a point-splitting regularization through a doubling of the basic observables. Another would be through approximating the basic point-localized observables by spatially smeared ones.

In \cite{OeZa:lcmeasure} an apparently different and simpler prescription for constructing probes to measure spacetime observables in QFT was presented. This modulus-square construction is based on factorizing a given observable as the modulus square of another observable and distributing the two factors between the two branches of the Schwinger-Keldysh double path integral. In the simplest case we consider an observable $\Lambda=A \overline{A}$, where $A$ is linear and a copy of $A$ is inserted in each of the two branches of the path integral. Apart from allowing full compositionality as taken advantage of in the present work as well, this is easy to calculate. What is more, a fully compositional renormalization scheme can be easily implemented for quadratic observables. A disadvantage is that apparently only probes to measure expectation values can be constructed and there is no accompanying notion of non-selective probe. However, it was envisaged already in \cite{OeZa:lcmeasure} to also consider observables that are weighted sum of modulus squares. Indeed, the construction presented in this work can be seen precisely as an instance of the modulus-square construction, as can be read off from expressions \eqref{eq:probeq} and \eqref{eq:stcomprprobe}. The general probes \eqref{eq:stopcalc} and \eqref{eq:probef} are then weighted integrals. What is more, choosing a weight function $f(q)=q$ (for a single observable) to obtain an expectation probe corresponds in terms of the modulus square construction precisely to the original observable
\begin{equation}
 \sqrt{2\pi}\epsilon \int_{-\infty}^{\infty}\xd q\, q |H^{\epsilon}_A(q)(\phi)|^2=A(\phi)
\end{equation}
as anticipated in the conclusions of \cite{OeZa:lcmeasure}. Nevertheless, an explicit comparison with the simpler probes constructed in that work would be interesting. However, this would require the introduction of a renormalization scheme, which is outside the scope of the present work.

As laid out in the introduction, this work is exclusively focused on the measurement process conceptualized as a fundamental ingredient of relativistic quantum theory. However, linking this to a von Neumann description in terms of an explicit modeling of a measurement device that is read out later is of considerable interest, not least when considering realization of measurement in experimental setups. The only remark that we venture to undertake in this direction is that we need an ancilla system with infinitely many degrees of freedom to match the continuous spectrum of the linear observables considered.

%% file: mathspec.tex

\section{On the regularized spectral decomposition}
\label{sec:mathspec}

Let $f:\R\to\C$ be Lebesgue measurable and essentially bounded, $\epsilon>0$. Define $f^\epsilon:\R\to\C$ by the following Gaussian convolution,
\begin{equation}
    f^\epsilon(q)\defeq \frac{1}{\sqrt{\pi}\epsilon}\int_{-\infty}^{\infty}\xd s f(s) \exp\left(-\frac{(q-s)^2}{\epsilon^2}\right) .
\end{equation}
\begin{lem}
   $f^\epsilon$ is Lebesgue measurable and essentially bounded, moreover
   $\|f^\epsilon\|_{\infty}\le\|f\|_{\infty}$.
\end{lem}
\begin{proof}
     This follows from a straightforward estimate.
\end{proof}

Let $A:L\to\R$ be a linear observable. Recall the spectral decomposition and its regularized version from Section~\ref{sec:instqft}.

\begin{lem}
    \begin{equation}
        \Pi_A[f^\epsilon]=\Pi_A^\epsilon[f] .
    \end{equation}
\end{lem}
\begin{proof}
    Since both operators exist and are bounded, it is enough to show coincidence of matrix elements on the dense subspace spanned by coherent states. Recall the notion of observable propagator from Sections~\ref{sec:amocor} and \ref{sec:sobs}. Treating $A$ as a slice observable we define $r\defeq \wuu(A,A)$. We use the following formulas for the sesquilinear form defined by the spectral measure and regularized spectral measure \cite{Oe:spectral}. (Note there are some minor notational differences, in particular $\|\xi\|^2=2r$ for $\xi$ as defined in \cite{Oe:spectral}.)
    \begin{align}
        B_A^\epsilon(\coh_\gamma,\coh_\beta)
        & =\langle\coh_\gamma,\Pi_A^\epsilon(q)\coh_\beta\rangle
        =\frac{1}{\sqrt{\pi (\epsilon^2 + 2 r)}}
        \langle\coh_\gamma, \coh_\beta\rangle
        \exp\left(-\frac{(A(\beta^- +\gamma^+)-q)^2}{\epsilon^2+2r}\right) , \\
        B_A(\coh_\gamma,\coh_\beta)
        & =\frac{1}{\sqrt{2\pi r}}
        \langle\coh_\gamma, \coh_\beta\rangle
        \exp\left(-\frac{(A(\beta^- +\gamma^+)-q)^2}{2r}\right) .
    \end{align}
    Note that the first formula is also a special case of \eqref{eq:gspecmatrix}. Then,
    \begin{align*}
        \langle \coh_\gamma,\Pi_A[f^\epsilon]\coh_\beta\rangle
        & =\int_{-\infty}^{\infty}\xd q\, f^\epsilon(q) B_A(q)(\coh_\gamma,\coh_\beta) \\
        & =\frac{1}{\sqrt{2r}\pi\epsilon}\int_{-\infty}^{\infty}\xd q \int_{-\infty}^{\infty}\xd s\, f(s) \langle\coh_\gamma, \coh_\beta\rangle
        \exp\left(-\frac{(q-s)^2}{\epsilon^2}-\frac{(A(\beta^- +\gamma^+)-q)^2}{2r}\right) \\
        & =\frac{1}{\sqrt{2r}\pi\epsilon}\int_{-\infty}^{\infty}\xd s\, f(s) \langle\coh_\gamma, \coh_\beta\rangle  \int_{-\infty}^{\infty}\xd q\,
        \exp\left(-\frac{(q-s)^2}{\epsilon^2}-\frac{(A(\beta^- +\gamma^+)-q)^2}{2r}\right) \\
        & =\frac{1}{\sqrt{\pi(\epsilon^2+2r)}}\int_{-\infty}^{\infty}\xd s\, f(s) \langle\coh_\gamma, \coh_\beta\rangle
        \exp\left(-\frac{(A(\beta^- +\gamma^+)-s)^2}{\epsilon^2 + 2r}\right) \\
        & =\int_{-\infty}^{\infty}\xd s\, f(s) B_A^\epsilon(\coh_\gamma,\coh_\beta)
        =\langle \coh_\gamma,\Pi_A^\epsilon[f]\coh_\beta\rangle . \tag*{\qedhere}
    \end{align*}
\end{proof}

%% file: mathst.tex

\section{Proofs in the spacetime framework}
\label{sec:mathst}

With definition \eqref{eq:aeobs} we recall the following identity.
\begin{equation}
  H_A^{\epsilon}(q)(\phi)=\frac{1}{\pi}\int_{-\infty}^{\infty}\xd t\, e^{-\epsilon^2 t^2} \exp\left(2\im t (A(\phi)-q)\right) .
  \label{eq:intid}
\end{equation}

\begin{lem}
    \label{lem:addid}
\begin{equation}
    \int_{-\infty}^{\infty}\xd s\, H_{A'}^{\epsilon'}(q-s) H_A^{\epsilon}(s)
    =H_{A+A'}^{\epsilon''}(q),\quad\text{with}\quad \epsilon''^2=\epsilon^2+\epsilon'^2 .
\end{equation}
\end{lem}
\begin{proof}
Fixing $\phi$ and thus $A(\phi)$ and $A'(\phi)$ this becomes a simple integral identity in $\R$.
\end{proof}

\begin{lem}
    \label{lem:did}
    If $f$ is continuous and increases less than the exponential of a square, then
    \begin{equation}
        H_A[f](\phi)=\lim_{\epsilon\to 0}H_A^\epsilon[f](\phi)=f(A(\phi)) .
        \label{eq:did}
    \end{equation}
\end{lem}
\begin{proof}
The proof is straightforward.
\end{proof}
Indeed, this identity corresponds exactly to the operator setting, where we have $\Pi_{\xi}[f]=f(\widehat{A_\xi})$ from the functional calculus of the spectral decomposition. Note that if $f$ is not continuous, but only measurable, then the identity \eqref{eq:did} is still valid almost everywhere. This is good enough for our purposes.

\begin{lem}
    Let $\Sigma_1\le\Sigma_2$ be spacelike hypersurfaces, $A,A':K_{[\Sigma_1,\Sigma_2]}\to\R$ linear observables and $\phi,\phi'\in L_{\Sigma_1}$. Then,
    \begin{multline}
        \discard_{\Sigma_2^\lf}\circ
        \cP_{[\Sigma_1,\Sigma_2]}[\exp(\im A)|\exp(\im A')](\Xi_{\phi|\phi'}) \\
        = \discard_{\Sigma_1^\lf}(\Xi_{\phi|\phi'})
        \exp\left(\im(A-A')(\phi^- +\phi'^+)\right) \exp\left(\wud(A,A')-\frac12\wuu(A,A)-\frac12\wdd(A',A')\right) .
    \end{multline}
\end{lem}
\begin{proof} This is adapted from Appendix~B.1 of \cite{OeZa:lcmeasure}.
\end{proof}

\begin{lem}
    \label{lem:discardint}
    Let $\Sigma_1\le\Sigma_2$ be spacelike hypersurfaces. Fix $\epsilon>0$, $s\in\R$, $A:K_{[\Sigma_1,\Sigma_2]}\to\R$ linear, $\sigma\in\cT_{\Sigma_1}$. Then:
    \begin{equation}
       \sqrt{2\pi}\epsilon\int_{-\infty}^{\infty}\xd q\;\discard_{\Sigma_2^\lf}\circ\cP_{[\Sigma_1,\Sigma_2]}[H_A^{\epsilon}(q+s)|H_A^{\epsilon}(q)](\sigma)
       =\exp\left(-\frac{s^2}{2\epsilon^2}\right)\discard_{\Sigma_1^\lf}(\sigma) .
    \end{equation}
\end{lem}
\begin{proof}
    It is sufficient to show this for $\sigma=\Xi_{\phi|\phi'}$. Thus,
    \begin{align}
        & \sqrt{2\pi}\epsilon\int_{-\infty}^{\infty}\xd q\;\discard_{\Sigma_2^\lf}\circ\cP_{[\Sigma_1,\Sigma_2]}[H_A^{\epsilon}(q+s)|H_A^{\epsilon}(q)](\Xi_{\phi|\phi'}) \label{eq:di1}\\
        & = \sqrt{\frac{2}{\pi^3}}\,\epsilon \int_{-\infty}^{\infty}\xd q \int_{-\infty}^{\infty}\xd t \int_{-\infty}^{\infty}\xd t' e^{-\epsilon^2 (t^2 + t'^2)-2\im q (t-t')-2\im s t} \nonumber \\
        & \qquad \discard_{\Sigma_2^\lf}\circ
        \cP_{[\Sigma_1,\Sigma_2]}[\exp(2\im t A)|\exp(2\im t' A)](\Xi_{\phi|\phi'}) \label{eq:di2} \\
        & = \discard_{\Sigma_1^\lf}(\Xi_{\phi|\phi'})
         \sqrt{\frac{2}{\pi^3}}\,\epsilon
        \int_{-\infty}^{\infty}\xd q \int_{-\infty}^{\infty}\xd t \int_{-\infty}^{\infty}\xd t' e^{-\epsilon^2 (t^2 + t'^2)-2\im q (t-t')-2\im s t} \nonumber \\
        & \qquad \exp\left(2\im (t-t') A(\phi^- + \phi'^+)\right)
        \exp\left(4 t t' \wh(A,A) -2t^2 \wuu(A,A) -2 t'^2 \wdd(A,A)\right) \label{eq:di3} \\
        & = \exp\left(-\frac{s^2}{2\epsilon^2}\right)\discard_{\Sigma_1^\lf}(\Xi_{\phi|\phi'}) .
    \end{align}
    In the first step from \eqref{eq:di1} to \eqref{eq:di2} we have applied the identity \eqref{eq:intid}. Subsequently, we have applied Lemma~\ref{lem:discardint} to arrive at expression \eqref{eq:di3}. For the final step of evaluating the integrals we recall properties of the observable propagators, compare Section~\ref{sec:amocor}. In particular, we have $\wdd(A,A)=\overline{\wuu(A,A)}$ and $\wh(A,A)=\Re(\wuu(A,A))$.
\end{proof}

\begin{proof}[Proof of Proposition~\ref{prop:nsprobeobs}]
    This follows as the special case $s=0$ of Lemma~\ref{lem:discardint}.
\end{proof}

\begin{figure}
  \centering
  \includegraphics[width=0.75\textwidth]{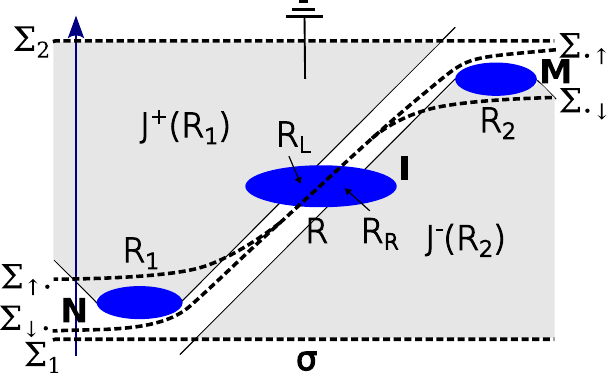}
\caption{Insertion of spacelike hypersurfaces into the setting of Figure~\ref{fig:caustrans} for the proof of causal transparency. Also, the region $R$ is in this manner divided into a left part $R_L$ and a right part $R_R$.}
\label{fig:caustrans2}
\end{figure}

\begin{proof}[Proof of Theorem~\ref{thm:ctst}]
  If either $R$ does not intersect the causal future of $R_1$ or does not intersect the causal past of $R_2$ the claim follows already from relativistic causality since the non-selective probe $\cA^{\epsilon}[\one]$ satisfies the causality axiom by Proposition~\ref{prop:nsprobeobs}. Without loss of generality we therefore suppose that there are spacelike hypersurfaces $\sdd,\sdu,\sud,\suu$ with the following properties, compare Figure~\ref{fig:caustrans2}.
  \begin{align}
    \Sigma_1\le\sdd\le R_1\le \sud\le R_2\le\suu\le\Sigma_2,\\ \Sigma_1\le\sdd\le\sdu\le R_1\le\suu\le\Sigma_2,\\
    \Sigma_1\le\sdd\le R_2\le\sdu\le\suu\le\Sigma_2,\\
    \suu\cap R=\sud\cap R=\sdu\cap R=\sdd\cap R .
  \end{align}
  Also, the hypersurfaces $\sdd,\sdu,\sud,\suu$ coincide in the region $R$ and divide it into a (left) future part $R_L$ and a (right) past part $R_R$.
  By additive decomposability (compare Section~\ref{sec:loccomp}) we split $A$ into linear observables $A_L:K_{R_L}\to \R$ and $A_R:K_{R_R}\to\R$ such that $A=A_L+A_R$ in the sense discussed. Set $\epsilon'\defeq \epsilon/\sqrt2$.
  \begin{align}
    &     \discard_{\Sigma_2^\lf}\circ\left(\cM\comp\cI\comp\cN\right)(\sigma) =\sqrt{2\pi}\epsilon \int_{\R}\xd q\;\discard_{\Sigma_2^\lf}\circ\left(\cM\comp\cP[H_A^\epsilon(q)|H_A^\epsilon(q)]\comp\cN\right)(\sigma) \label{eq:plc1}\\
    & =\sqrt{2\pi}\epsilon \int_{\R^3}\xd q\,\xd s\,\xd t\;\discard_{\Sigma_2^\lf}\circ\left(\cM\comp
    \cP[H_{A_L}^{\epsilon'}(q-s) H_{A_R}^{\epsilon'}(s)|H_{A_L}^{\epsilon'}(q-t) H_{A_R}^{\epsilon'}(t)]\comp\cN\right)(\sigma) \label{eq:plc2}\\
    & =\sqrt{2\pi}\epsilon \int_{\R^3}\xd q\,\xd s\,\xd t\;
    \discard_{\Sigma_2^\lf}\circ
    \cP_{[\suu,\Sigma_2]}[H_{A_L}^{\epsilon'}(q-s) | H_{A_L}^{\epsilon'}(q-t)]
    \circ\cN_{[\sdu,\suu]}\nonumber\\
    &\qquad \circ\cM_{[\sdd,\sdu]}
    \circ\cP_{[\Sigma_1,\sdd]}[H_{A_R}^{\epsilon'}(s) | H_{A_R}^{\epsilon'}(t)](\sigma) \label{eq:plc3}\\
    & =\sqrt{2\pi}\epsilon \int_{\R^2}\xd s\,\xd t\, \exp\left(-\frac{(s-t)^2}{2{\epsilon'}^2}\right)
    \discard_{\suu^\lf}
    \circ\cN_{[\sdu,\suu]}\nonumber\\
    &\qquad \circ\cM_{[\sdd,\sdu]}
    \circ\cP_{[\Sigma_1,\sdd]}[H_{A_R}^{\epsilon'}(s) | H_{A_R}^{\epsilon'}(t)](\sigma) \label{eq:plc4}\\
    & =\sqrt{2\pi}\epsilon \int_{\R^2}\xd s\,\xd t\, \exp\left(-\frac{(s-t)^2}{2{\epsilon'}^2}\right)
    \discard_{\sdu^\lf}\circ\cM_{[\sdd,\sdu]}
    \circ\cP_{[\Sigma_1,\sdd]}[H_{A_R}^{\epsilon'}(s) | H_{A_R}^{\epsilon'}(t)](\sigma) \label{eq:plc5}\\
    & =\sqrt{2\pi}\epsilon \int_{\R^3}\xd q\,\xd s\,\xd t\;
    \discard_{\Sigma_2^\lf}\circ
    \cP_{[\sdu,\Sigma_2]}[H_{A_L}^{\epsilon'}(q-s) | H_{A_L}^{\epsilon'}(q-t)]
    \nonumber\\
    &\qquad \circ\cM_{[\sdd,\sdu]}
    \circ\cP_{[\Sigma_1,\sdd]}[H_{A_R}^{\epsilon'}(s) | H_{A_R}^{\epsilon'}(t)](\sigma) \label{eq:plc6}\\
    & =\sqrt{2\pi}\epsilon \int_{\R^3}\xd q\,\xd s\,\xd t\;
    \discard_{\Sigma_2^\lf}\circ
    \left(\cM
    \comp\cP[H_{A_L}^{\epsilon'}(q-s) H_{A_R}^{\epsilon'}(s) | H_{A_L}^{\epsilon'}(q-t) H_{A_R}^{\epsilon'}(t)]\right)(\sigma) \label{eq:plc7}\\
    & =\sqrt{2\pi}\epsilon \int_{\R}\xd q\;\discard_{\Sigma_2^\lf}\circ\left(\cM\comp\cP[H_A^\epsilon(q)|H_A^\epsilon(q)]\right)(\sigma)
    \label{eq:plc8}\\
    & = \discard_{\Sigma_2^\lf}\circ\left(\cM\comp\cI\right)(\sigma) .
  \end{align}
  In the step from equation \eqref{eq:plc1} to \eqref{eq:plc2} we have used the identity of Lemma~\ref{lem:addid} twice, and correspondingly in the step from equation \eqref{eq:plc7} to equation \eqref{eq:plc8}. In the step from equation \eqref{eq:plc3} to \eqref{eq:plc4} we have used Lemma~\ref{lem:discardint}. We have used it again in the step from equation \eqref{eq:plc5} to \eqref{eq:plc6}. In the step from equation \eqref{eq:plc4} to \eqref{eq:plc5} we have used the causality property \eqref{eq:nrcaus} of the non-selective probe $\cN$.
\end{proof}

%% file: pmultiobs.tex

\section{Proofs for composite observables}
\label{sec:pmultiobs}

\begin{proof}[Proof of Proposition~\ref{prop:multicor}]
    Suppose that the matrix $\epsilon^2\id +2W$ is not invertible. Let $(\lambda_1,\ldots,\lambda_n)$ be a non-zero vector in its kernel. We then have in particular, that
    \begin{equation}
     \sum_{i=1}^n \epsilon^2 |\lambda_i|^2
     +2\sum_{i,j=1}^n\overline{\lambda_i} R_{i j}\lambda_j
     +2\im\sum_{i,j=1}^n\overline{\lambda_i} I_{i j}\lambda_j
     =0 .
    \end{equation}
    The first term is obviously real, the second term is real since $R$ is positive and the last term is imaginary since $I_{i j}$ is real and $I_{i j}=I_{j i}$. Thus, the sum of the first two terms alone has to vanish. By positivity of $R$ the second term is larger or equal to zero as is the first term. Thus, both have to be zero. This is only possible if $\epsilon=0$ and $R$ is degenerate (i.e.\ not invertible).

    We use the identity \eqref{eq:intid}, the factorization identity \eqref{eq:corfact}, the expression \eqref{eq:vev} of the vacuum correlation function, as well as the notation $\va{t}=(t_1,\ldots,t_n)$, to obtain,
    \begin{align}
      & \rho[H_{A_1}^{\epsilon}(q)\cdots H_{A_n}^{\epsilon}(q)](\ncoh_{\phi}) \nonumber \\
      & = \frac{1}{\pi^n} \int_{\R^n}\xd t_1\cdots\xd t_n\, 
      e^{-\epsilon^2 \va{t}^2}
      \rho\left[\exp\left(\sum_{k=1}^n 2\im t_k (A_k(\cdot)-q_k)\right)\right](\ncoh_{\phi}) \nonumber \\
      & =\frac{1}{\pi^n}\rho(\ncoh_{\phi}) \int_{\R^n}\xd t_1\cdots\xd t_n\,
      \exp\left(2\im\, \va{t}\cdot\left(\va{\phi}-\va{q}\right)\right)
      \exp\left(-\va{t}^\mathrm{T} (\epsilon^2\id+2W)\,\va{t}\right)
      \label{eq:multicor1} \\
      & =\frac{1}{\sqrt{\pi^n \det(\epsilon^2\id + 2W)}}\rho(\ncoh_{\phi})
      \exp\left(- \left(\va{\phi}-\va{q}\right)^\mathrm{T}
       \left(\epsilon^2\id+2W\right)^{-1}
      \left(\va{\phi}-\va{q}\right) \right) . \nonumber \tag*{\qedhere}
    \end{align}
\end{proof}

\begin{proof}[Proof of Proposition~\ref{prop:mprobeweyl}]
    We perform the Gaussian integral using the identity
    \begin{equation}
       \left(\epsilon^2\id+2W\right)^{-1}
       +\left(\epsilon^2\id+2\overline{W}\right)^{-1}
       =2\left(\epsilon^2\id+2W\right)^{-1}
        \left(\epsilon^2\id+2R\right) 
       \left(\epsilon^2\id+2\overline{W}\right)^{-1} .
    \end{equation}
    This is,
    \begin{align*}
    & \cA^\epsilon[w_{\va{s}}](\Xi_{\beta|\gamma})
     =\int_{\R^n}\xd q_1\cdots\xd q_n\, w_{\va{s}}(\va{q}) \cA^\epsilon[\va{q}](\Xi_{\beta|\gamma}) \\
    & =\np(\Xi_{\beta|\gamma})\sqrt{\frac{2}{\pi}}^n
    \frac{\epsilon^n}{|\det(\epsilon^2\id + 2W)|} \int_{\R^n}\xd q_1\cdots\xd q_n\, \exp\left(\im \va{s}\cdot\va{q}\right)\\
    & \quad  \exp\left(- \left(\va{\beta}-\va{q}\right)^\mathrm{T}
       \left(\epsilon^2\id+2W\right)^{-1}
      \left(\va{\beta}-\va{q}\right) 
      -\left(\overline{\va{\gamma}}-\va{q}\right)^\mathrm{T}
       \left(\epsilon^2\id+2\overline{W}\right)^{-1}
      \left(\overline{\va{\gamma}}-\va{q}\right) 
      \right) \\
    & =\np(\Xi_{\beta|\gamma})\sqrt{\frac{2}{\pi}}^n
    \frac{\epsilon^n}{|\det(\epsilon^2\id + 2W)|} \int_{\R^n}\xd q_1\cdots\xd q_n\, \exp\left(\im \va{s}\cdot\va{q}\right)\\
    & \quad  \exp\left(- \va{\beta}^\mathrm{T}
       \left(\epsilon^2\id+2W\right)^{-1}
       \va{\beta} 
      -\overline{\va{\gamma}}^\mathrm{T}
       \left(\epsilon^2\id+2\overline{W}\right)^{-1} \overline{\va{\gamma}}
      \right) \\
    & \quad \exp\left(2\va{\beta}^\mathrm{T}
       \left(\epsilon^2\id+2W\right)^{-1}\va{q} 
      +2\overline{\va{\gamma}}^\mathrm{T}
       \left(\epsilon^2\id+2\overline{W}\right)^{-1}\va{q} 
      \right) \\
    & \quad  \exp\left(- \va{q}^\mathrm{T}
       \left(\left(\epsilon^2\id+2W\right)^{-1}
       +\left(\epsilon^2\id+2\overline{W}\right)^{-1}\right)
      \va{q} 
      \right) \\
    & =\np(\Xi_{\beta|\gamma})\frac{\epsilon^n}{\sqrt{\det(\epsilon^2\id + 2R)}}
      \exp\left(- \va{\beta}^\mathrm{T}
       \left(\epsilon^2\id+2W\right)^{-1}
       \va{\beta} 
      -\overline{\va{\gamma}}^\mathrm{T}
       \left(\epsilon^2\id+2\overline{W}\right)^{-1} \overline{\va{\gamma}}
      \right) \\
    & \quad  \exp\left(\frac12\va{\beta}^\mathrm{T}
       \left(\epsilon^2\id+2W\right)^{-1}
       \left(\epsilon^2\id+2\overline{W}\right)
       \left(\epsilon^2\id+2R\right)^{-1}
       \va{\beta}\right) \\
    & \quad \exp\left(\frac12\overline{\va{\gamma}}^\mathrm{T}
       \left(\epsilon^2\id+2\overline{W}\right)^{-1}
       \left(\epsilon^2\id+2W\right)
       \left(\epsilon^2\id+2R\right)^{-1}
       \overline{\va{\gamma}}\right)
       \exp\left(\va{\beta}^\mathrm{T}
       \left(\epsilon^2\id+2R\right)^{-1}
       \overline{\va{\gamma}}\right) \\
    & \quad \exp\left(\frac{\im}{4}\va{s}^\mathrm{T}
       \left(\epsilon^2\id+2\overline{W}\right)
       \left(\epsilon^2\id+2R\right)^{-1}\va{\beta}\right)
       \exp\left(\frac{\im}{4}\va{s}^\mathrm{T}
       \left(\epsilon^2\id+2W\right)
       \left(\epsilon^2\id+2R\right)^{-1}\overline{\va{\gamma}}\right) \\
    & \quad \exp\left(-\frac18\va{s}^\mathrm{T}
       \left(\epsilon^2\id+2\overline{W}\right)
       \left(\epsilon^2\id+2R\right)^{-1}
       \left(\epsilon^2\id+2W\right)\va{s}\right) \\
    & =\np(\Xi_{\beta|\gamma})
       \frac{\epsilon^n}{\sqrt{\det(\epsilon^2\id + 2R)}}
       \exp\left(-\frac12 \left(\va{\beta}-\overline{\va{\gamma}}\right)^{\mathrm{T}}
       \left(\epsilon^2\id+2R\right)^{-1}
       \left(\va{\beta}-\overline{\va{\gamma}}\right)
       \right)\\
    & \quad \exp\left(\im\, \va{s}\cdot\frac{\va{\beta}
       +\overline{\va{\gamma}}}{4}+\va{s}^\mathrm{T}  I \left(\epsilon^2\id+2R\right)^{-1}\frac{\va{\beta}-\overline{\va{\gamma}}}{2}\right)
       \exp\left(-\frac18 \va{s}^{\mathrm{T}}\left(\epsilon^2\id+2R
       - 4 I\left(\epsilon^2\id+2R\right)^{-1} I\right) \va{s}\right) .
    \end{align*}
\end{proof}

\begin{proof}[Proof of Proposition~\ref{prop:mprobeqdisc}]
  In the present context of evolution in the region $[\Sigma_1,\Sigma_2]$ we can write the amplitude between coherent states as the inner product \eqref{eq:cohip},
  \begin{equation}
    \rho(\coh_{\phi_1}\tens\coh_{\phi_2})
    =\langle \coh_{\phi_2},\coh_{\phi_1}\rangle
    =\exp\left(\frac12\{\phi_1,\phi_2\}\right) .
  \end{equation}
  Here $\{\cdot,\cdot\}$ is the phase space inner product, and we identify the solution space $L_{[\Sigma_1,\Sigma_2]}$ as usual with the phase space $L$. What is more, to each linear observable $A_k$ there exists a unique corresponding phase space element $\tau_k\in L$ such that for any complexified solution $\phi\in L_{[\Sigma_1,\Sigma_2]}^\C$ we have the identity \cite{CoOe:locgenvac,OeZa:lcmeasure}
  \begin{equation}
  A_k(\phi)= 2\omega(\tau_k,\phi) .
  \end{equation}
  We recall that we may rewrite the symplectic structure in terms of the phase space inner product,
  \begin{equation}
  2\omega(\tau_k,\phi)=\frac{\im}{2}\{\phi^-,\tau_k\}-\frac{\im}{2}\{\tau_k,\phi^+\} ,
  \end{equation}
  and have the propagator identity
  \begin{equation}
  \wud(A_k,A_l)=\frac12\{\tau_k,\tau_l\} .
  \end{equation}
  We set $K_{i j}\defeq \wud(A_i,A_j)$ and $K=R+\im J$.
  It is now convenient to start the derivation with the use of expression \eqref{eq:multicor1}.
  \begin{align*}
  & \discard\circ\cA^\epsilon[\va{q}](\Xi_{\beta|\gamma})
   =(\sqrt{2\pi}\epsilon)^n\frac{1}{\pi^{2n}} \int_{\hat{L}}\xd\nu(\phi)\,
   \exp\left(\frac12\{\beta,\phi\}+\frac12\{\phi,\gamma\}\right) \\
  & \quad \int_{\R^n}\xd t_1\cdots\xd t_n\,
      \exp\left(2\im\, \va{t}\cdot\left(\frac{\im}{2}\{\beta^-,\va{\tau}\}-\frac{\im}{2}\{\va{\tau},\phi^+\}-\va{q}\right)\right)
      \exp\left(-\va{t}^\mathrm{T} (\epsilon^2\id+2W)\,\va{t}\right) \\
  & \quad \int_{\R^n}\xd t'_1\cdots\xd t'_n\,
      \exp\left(-2\im\, \va{t}'\cdot\left(-\frac{\im}{2}\{\va{\tau},\gamma^+\}+\frac{\im}{2}\{\phi^-,\va{\tau}\}-\va{q}\right)\right)
      \exp\left(-\va{t}'{}^\mathrm{T} (\epsilon^2\id+2\overline{W})\,\va{t}'\right) \\
  & =(\sqrt{2\pi}\epsilon)^n\frac{1}{\pi^{2n}}
   \int_{\R^{2n}}\xd t_1\cdots\xd t_n\, \xd t'_1\cdots\xd t'_n\,
   \exp\left(\frac12\{\beta+2\va{t}\cdot\va{\tau},\gamma+2\va{t}'\cdot\va{\tau}\}\right) \\
  & \quad \exp\left(-\frac12\{\beta,2\va{t}\cdot\va{\tau}\}-2\im \va{t}\cdot\va{q}\right)\exp\left(-\va{t}^\mathrm{T} (\epsilon^2\id+2W)\,\va{t}\right) \\
  & \quad\exp\left(-\frac12\{2 \va{t}'\cdot\va{\tau},\gamma\}+2\im \va{t}'\cdot\va{q}\right) \exp\left(-\va{t}'{}^\mathrm{T} (\epsilon^2\id+2\overline{W})\,\va{t}'\right) \\
  & =\discard(\Xi_{\beta|\gamma})(\sqrt{2\pi}\epsilon)^n\frac{1}{\pi^{2n}}
   \int_{\R^{2n}}\xd t_1\cdots\xd t_n\, \xd t'_1\cdots\xd t'_n\,
   \exp\left(\frac12\{2\va{t}\cdot\va{\tau},2\va{t}'\cdot\va{\tau}\}\right) \\
  & \quad \exp\left(-\frac12\{\beta,2(\va{t}-\va{t}')\cdot\va{\tau}\}-2\im \va{t}\cdot\va{q}\right)\exp\left(-\va{t}^\mathrm{T} (\epsilon^2\id+2W)\,\va{t}\right) \\
  & \quad\exp\left(\frac12\{2 (\va{t}-\va{t}')\cdot\va{\tau},\gamma\}+2\im \va{t}'\cdot\va{q}\right) \exp\left(-\va{t}'{}^\mathrm{T} (\epsilon^2\id+2\overline{W})\,\va{t}'\right) \\
  & =\discard(\Xi_{\beta|\gamma})(\sqrt{2\pi}\epsilon)^n\frac{1}{\pi^{2n}}
   \int_{\R^{2n}}\xd t_1\cdots\xd t_n\, \xd t'_1\cdots\xd t'_n\,
   \exp\left(-\,\va{t}^{\mathrm{T}} (-2K) \va{t}'-\,\va{t}'{}^{\mathrm{T}} (-2\overline{K}) \va{t}\right) \\
  & \quad \exp\left(2\im (\va{t}-\va{t}')\cdot(\va{\beta}^- + \va{\gamma}^+-\va{q})\right)
  \exp\left(-\va{t}^\mathrm{T} (\epsilon^2\id+2W)\,\va{t}\right)
  \exp\left(-\va{t}'{}^\mathrm{T} (\epsilon^2\id+2\overline{W})\,\va{t}'\right) .
  \end{align*}
  To simplify this we perform a change of variables as $\va{s}=\va{t}+\va{t}'$ and $\va{d}=\va{t}-\va{t'}$. We also recall $W_{\mathrm{ret}}=I+J$ and $W_{\mathrm{adv}}=I-J$. This yields,
  \begin{align*}
  & \discard\circ\cA^\epsilon[\va{q}](\Xi_{\beta|\gamma})
  =\discard(\Xi_{\beta|\gamma})
   \left(\sqrt{\frac{\pi}{2}}\epsilon\right)^n\frac{1}{\pi^{2n}}
   \int_{\R^{2n}}\xd s_1\cdots\xd s_n\, \xd d_1\cdots\xd d_n \\
  & \quad \exp\left(2\im\, \va{d}\cdot(\va{\beta}^- + \va{\gamma}^+-\va{q})\right)
  \exp\left(-\va{s}^\mathrm{T} \frac{\epsilon^2}{2}\id \va{s}
  -\va{d}^{\mathrm{T}} \left(\frac{\epsilon^2}{2}\id+2R\right)\,\va{d}
  -\va{s}^{\mathrm{T}} \im W_{\mathrm{ret}}\va{d}-\va{d}^{\mathrm{T}} \im W_{\mathrm{adv}}\va{s}
  \right) \\
  & =\discard(\Xi_{\beta|\gamma})\frac{1}{\pi^{n}}
  \int_{\R^n}\xd d_1\cdots\xd d_n
  \exp\left(2\im\, \va{d}\cdot(\va{\beta}^- + \va{\gamma}^+-\va{q})\right) \\
  & \quad 
  \exp\left(
  -\va{d}^{\mathrm{T}} \left(\frac{\epsilon^2}{2}\id+2R+\frac{2}{\epsilon^2}W_{\mathrm{adv}}W_{\mathrm{ret}}\right)\,\va{d}
  \right) \\
  & =\discard(\Xi_{\beta|\gamma})\frac{1}{\sqrt{\pi^{n}\det(X_\epsilon)}}
  \exp\left(-(\va{\beta}^- + \va{\gamma}^+-\va{q})^{\mathrm{T}}X_{\epsilon}^{-1}(\va{\beta}^- + \va{\gamma}^+-\va{q})\right) . \tag*{\qedhere}
  \end{align*}
\end{proof}

%% file: main.bbl
\begin{thebibliography}{10}
\providecommand{\url}[1]{\texttt{#1}}
\providecommand{\urlprefix}{URL }
\providecommand{\selectlanguage}[1]{\relax}
\providecommand{\eprint}[2][]{\url{#2}}

\bibitem{ArHa:collisionlocalobs}
H.~Araki, R.~Haag, \textit{Collision cross sections in terms of local observables}, Comm. Math. Phys. \textbf{4} (1967) 77--91.

\bibitem{Oe:spectral}
R.~Oeckl, \textit{Spectral decomposition of field operators and causal measurement in quantum field theory}, J. Math. Phys. \textbf{66} (2025) 042302, \eprint{2409.08748}.

\bibitem{BaLaPr:marcocontobsqm}
A.~Barchielli, L.~Lanz, G.~M. Prosperi, \textit{A model for the macroscopic description and continual observations in quantum mechanics}, Nuovo Cimento \textbf{72} (1982) 79--121.

\bibitem{GhRiWe:micromacro}
G.~C. Ghirardi, A.~Rimini, T.~Weber, \textit{Unified dynamics for microscopic and macroscopic systems}, Phys. Rev. \textbf{D 34} (1986) 470--491.

\bibitem{BaLaPr:opvalstochastic}
A.~Barchielli, L.~Lanz, G.~M. Prosperi, \textit{Statistics of continuous trajectories in quantum mechanics: Operation-valued stochastic processes}, Found. Phys. \textbf{13} (1983) 779--812.

\bibitem{Sor:impossible}
R.~Sorkin, \textit{Impossible Measurements on Quantum Fields}, Directions in General Relativity, (eds. B.~L. Hu, T.~A. Jacobson), Cambridge University Press, Cambridge, 1993, pp. 293--305, \eprint{gr-qc/9302018}.

\bibitem{BoJuKe:impossible}
L.~Borsten, I.~Jubb, G.~Kells, \textit{Impossible measurements revisited}, Phys. Rev. \textbf{D 104} (2021) 025012, \eprint{1912.06141}.

\bibitem{Jub:causalupdates}
I.~Jubb, \textit{Causal state updates in real scalar quantum field theory}, Phys. Rev. \textbf{D 105} (2022) 025003, \eprint{2106.09027}.

\bibitem{Oe:dmf}
R.~Oeckl, \textit{A Positive Formalism for Quantum Theory in the General Boundary Formulation}, Found. Phys. \textbf{43} (2013) 1206--1232, \eprint{1212.5571}.

\bibitem{Oe:posfound}
R.~Oeckl, \textit{A local and operational framework for the foundations of physics}, Adv. Theor. Math. Phys. \textbf{23} (2019) 437--592, \eprint{1610.09052v3}.

\bibitem{Kel:diagramne}
L.~V. Keldysh, \textit{Diagram technique for nonequilibrium processes}, Zh. Eksp. Teor. Fiz. \textbf{47} (1964) 1515--1527.

\bibitem{OeZa:lcmeasure}
R.~Oeckl, A.~Zampeli, \textit{Towards local and compositional measurements in quantum field theory}, \eprint{2505.10968}.

\bibitem{vne:mathgrundquant}
J.~von Neumann, \textit{Mathematische Grundlagen der Quantenmechanik}, Springer, Berlin, 1932.

\bibitem{HeKr:opmeasureii}
K.-E. Hellwig, K.~Kraus, \textit{Operations and measurements. {II}}, Commun. Math. Phys. \textbf{16} (1970) 142--147.

\bibitem{FeVe:qftlocalmeasure}
C.~J. Fewster, R.~Verch, \textit{Quantum Fields and Local Measurements}, Commun. Math. Phys. \textbf{378} (2020) 851--889, \eprint{1810.06512}.

\bibitem{BiDa:qftcurved}
N.~D. Birrell, P.~C.~W. Davies, \textit{Quantum Fields in Curved Space}, Cambridge University Press, Cambridge, 1982.

\bibitem{PGGaMM:detectormeasurementqft}
J.~Polo-Gómez, L.~J. Garay, E.~Martín-Martínez, \textit{A detector-based measurement theory for quantum field theory}, Phys. Rev. \textbf{D 105} (2022) 065003, \eprint{2108.02793}.

\bibitem{PaFr:eliminatingimpossible}
M.~Papageorgiou, D.~Fraser, \textit{Eliminating the "{}impossible": Recent progress on local measurement theory for quantum field theory}, Found. Phys. \textbf{54} (2024) 26, \eprint{2307.08524}.

\bibitem{AlJu:measurecausal}
E.~Albertini, I.~Jubb, \textit{Are Ideal Measurements of Real Scalar Fields Causal?}, \eprint{2306.12980}.

\bibitem{MaNa:qftfv}
J.~Mandrysch, M.~Navascués, \textit{Quantum field measurements in the Fewster-Verch framework}, Lett. Math. Phys. \textbf{115} (2025) 115, \eprint{2411.13605}.

\bibitem{Oe:holomorphic}
R.~Oeckl, \textit{Holomorphic Quantization of Linear Field Theory in the General Boundary Formulation}, SIGMA \textbf{8} (2012) 050, \eprint{1009.5615}.

\bibitem{Oe:feynobs}
R.~Oeckl, \textit{Schr\"odinger-Feynman quantization and composition of observables in general boundary quantum field theory}, Adv. Theor. Math. Phys. \textbf{19} (2015) 451--506, \eprint{1201.1877}.

\bibitem{Oe:GBQFT}
R.~Oeckl, \textit{General boundary quantum field theory: Foundations and probability interpretation}, Adv. Theor. Math. Phys. \textbf{12} (2008) 319--352, \eprint{hep-th/0509122}.

\bibitem{Jac:schroedinger}
R.~Jackiw, \textit{Analysis of infinite-dimensional manifolds---Schr\"odinger representation for quantized fields}, Field theory and particle physics (Campos do Jord\~ao, 1989), World Scientific, River Edge, 1990, pp. 78--143.

\bibitem{Hat:qft}
B.~Hatfield, \textit{Quantum Field Theory of Point Particles and Strings}, Addison-Wesley, Redwood City, 1992.

\bibitem{CoOe:locgenvac}
D.~Colosi, R.~Oeckl, \textit{Locality and General Vacua in Quantum Field Theory}, SIGMA \textbf{17} (2021) 073, \eprint{2009.12342}.

\bibitem{CoOe:vaclag}
D.~Colosi, R.~Oeckl, \textit{The Vacuum as a Lagrangian subspace}, Phys. Rev. \textbf{D 100} (2019) 045018, \eprint{1903.08250}.

\bibitem{Sor:qmeasure}
R.~D. Sorkin, \textit{Quantum mechanics as quantum measure theory}, Mod. Phys. Lett. A \textbf{09} (1994) 3119--3127, \eprint{gr-qc/9401003}.

\bibitem{Hol:diviquantumprob}
A.~S. Kholevo, \textit{Infinitely Divisible Measurements in Quantum Probability Theory}, Theory Probab. Appl. \textbf{31} (1987) 493--497.

\bibitem{HaKa:aqft}
R.~Haag, D.~Kastler, \textit{An Algebraic Approach to Quantum Field Theory}, J. Math. Phys. \textbf{5} (1964) 848--861.

\bibitem{Haa:lqp}
R.~Haag, \textit{Local Quantum Physics}, Springer, Berlin, 1992.

\end{thebibliography}
